\documentclass[11pt]{article}

\usepackage{booktabs} 
\usepackage[ruled]{algorithm2e} 

\SetAlFnt{\small}
\SetAlCapFnt{\small}
\SetAlCapNameFnt{\small}
\SetAlCapHSkip{0pt}
\IncMargin{-\parindent}

 \usepackage[margin=1in]{geometry}

\usepackage{bbding}
\usepackage[nointegrals]{wasysym}
\usepackage[utf8]{inputenc}
\usepackage{xspace}
\usepackage{mathtools}
\usepackage{url}
\usepackage{multirow,multicol}
\usepackage{subcaption}
\usepackage{tikz}
\usepackage{hyperref}
\usetikzlibrary{arrows, fit, positioning}
\usetikzlibrary{backgrounds,automata,calc}
\usetikzlibrary{decorations.pathreplacing}
\usepackage{csquotes}
\usepackage{amsmath,amsfonts,amsthm}
\allowdisplaybreaks
\usepackage{bm,bbm}
\usepackage{pgfplots}
\pgfplotsset{compat=1.17}

\usepackage{enumerate}
\usepackage{paralist}
\usepackage[shortlabels]{enumitem}
\usepackage{appendix}
\usepackage[capitalize]{cleveref}
\usepackage{thmtools}
\usepackage{mathtools}
\usepackage[numbers]{natbib}
\usepackage{changepage}
 \usepackage{authblk}
\usepackage{commath}
\usepackage{makecell}

\newcommand{\Inf}{\text{Inf}\xspace}
\newcommand{\Nbd}{\text{Nbd}\xspace}

\newcommand{\USW}{\texttt{USW}\xspace}
\newcommand{\NSW}{\texttt{NSW}\xspace}

\newcommand{\OPT}{\texttt{OPT}\xspace}

\newcommand{\vare}{\varepsilon}

\newcommand{\R}{\mathbb{R}}

\DeclareMathOperator*{\E}{\mathbb{E}}

\renewcommand{\cal}[1]{\mathcal{#1}}

\usepackage{mathtools}

\setlist{topsep=0.5ex,itemsep=0.1ex}

\newtheorem{theorem}{Theorem}[section]
\newtheorem{prob}{Problem Instance Family}
\newtheorem{lemma}[theorem]{Lemma}
\newtheorem{conj}[theorem]{Conjecture}
\newtheorem{claim}[theorem]{Claim}

\newtheorem*{theorem*}{Theorem}

\theoremstyle{definition}
\newtheorem{remark}[theorem]{Remark}
\newtheorem{definition}[theorem]{Definition}
\newtheorem{obs}[theorem]{Observation}

\title{Improved Hardness Results for Nash Social Welfare, Budgeted Allocation and GAP via the Unique Games Conjecture}

\author{Vignesh Viswanathan}

\affil{University of Massachusetts, Amherst, USA \\ \texttt{vviswanathan@umass.edu}}
\date{}

\begin{document}

\maketitle

\begin{abstract}
We consider the problem of dividing a set of indivisible goods among agents with additive valuations. This problem has been studied under various objectives in both the computer science and the operations research literature. Our main contribution is a novel dictator test using this problem, which can separate a dictator from any function sufficiently far from a dictator. We use this test to prove the following hardness results (assuming the unique games conjecture is true):
\begin{itemize}
    \item[-] We show that it is NP-hard to approximate the max Nash welfare by a factor better than $\sqrt[3]{\frac{81}{65}} - \varepsilon \approx 1.0761$. This improves on the previous best known inapproximability factor of $\sqrt{\frac87} - \varepsilon \approx 1.069$.
    \item[-] We show that it is NP-hard to approximate the maximum budgeted allocation by a factor better than $\frac{243}{227} - \varepsilon \approx 1.07$. This improves on the previous best known inapproximability factor of $\frac{16}{15} - \varepsilon \approx 1.067$. 
    \item[-] We show that it is NP-hard to approximate the max generalized assignment problem (GAP) by a factor better than $\frac{145}{129} - \varepsilon \approx 1.124$. This improves on the previous best known inapproximability factor of $\frac{11}{10} - \varepsilon \approx 1.10$.
\end{itemize}
\end{abstract}
%
%
%
%

\newpage
\tableofcontents
\newpage

\section{Introduction}
We consider the problem of dividing a set of indivisible goods $G$ among a set of agents $N$ who have differing preferences over these goods. These preferences are captured using an additive valuation function $v_i: 2^G \rightarrow \R_{\ge 0}$ which describes the value agent $i$ has for each bundle of goods $S \subseteq G$. The goal of the problem is to find an allocation $X = (X_1, \dots, X_n)$, which is an $n$-subpartition of the goods, that maximizes some objective function. This is a fundamental problem in the EconCS community, which has been studied under several different objectives. In this paper, we consider the following well-studied objectives:

\noindent \textbf{Max Nash Welfare}\cite{Caragiannis2019MNW,Anari2017MNW,Barman2018FindingFA,garg2018budgetadditive,lee:j:apx-hardness,Cole2017nashwelfare,cole2015nash}: An allocation $X$ maximizes Nash welfare if it maximizes the geometric mean of agent utilities $\left (\prod_{i \in N} v_i(X_i) \right )^{\frac{1}{|N|}}$.

\noindent \textbf{Max Budgeted Allocation}\cite{Azar2008BudgetedAllocation,Srinivasan2008Budgeted,Garg2001BudgetedAllocations,Andelman2004BudgetedAllocations,Chakrabarty2010BudgetAdditive,Kalaitzis2016ImprovedBudgetedAllocation,Kalaitzis2015ConfigurationLPBudgetedAllocation}: Given agent budgets $b_i$ for each $i \in N$, an allocation $X$ is a max budgeted allocation if $X$ maximizes $\sum_{i \in N} \min\{b_i, v_i(X_i)\}$.

\noindent \textbf{Max Generalized Assignment Problem (GAP)}\cite{Fleischer2006GAP,Feige2006GAPImprovement,ShmoysTardos1993GAP,Oncan2007GAPSurvey,Morales2004GAPSurvey}: Given good sizes $s_g$ for each good $g \in G$ and agent capacities $c_i$ for each $i \in N$, the goal of GAP is to find an allocation $X$ that maximizes $\sum_{i \in N} v_i(X_i)$ subject to the bundle $X_i$ having total size at most $c_i$ for each $i \in N$.

These objectives are NP-hard to maximize exactly. The central research question has therefore been, `what is the best approximation ratio achievable?'. For each of the objectives described above, there is a long line of work contributing approximation algorithms. 
In contrast, very few papers study computational lower bounds. 
In this paper, we approach the central research question from a computational hardness perspective.

We prove improved inapproximability results for all three objectives described above when agents have additive valuations. The previous best inapproximability results for all the problems we consider follow from reductions using the Max-E3-Lin-2 problem \cite{Hastad2001OptimalInapproximability}. We improve on these results using long code testing. Long code testing is a technique that uses problem instances to test whether a boolean function satisfies a specific property; for example, whether the boolean function is a dictator. These tests can be converted to hardness results either using the PCP theorem \cite{Arora1998ProofVerification,AroraSafra1998PCP} or using the Unique Games Conjecture \cite{Khot2002UniqueGames}. Several breakthrough hardness results follow from a long code test; examples include max cut \cite{Khot2007MaxCutUGC}, vertex cover \cite{KhotRegev2008VertexCover,Bansal2009OneFreeBit} and Max-E3-Lin-2 \cite{Hastad2001OptimalInapproximability}. Our work is the first to employ this technique to prove hardness results for allocation problems.

\subsection{Main Results}

We prove the following three results in this paper. These results are summarized in Table \ref{tab:results}.

\begin{restatable}{theorem}{thmnash}\label{thm:nash}
For any constant $\vare > 0$, assuming the unique games conjecture, it is NP-hard to approximate the max Nash welfare by a factor better than $\sqrt[3]{\frac{81}{65}} - \vare \approx 1.0761$ when agents have additive valuations.
\end{restatable}

\paragraph{Related Work} The best known approximation algorithm for the problem achieves approximation ratio $e^{1/e} \approx 1.45$ \cite{Barman2018FindingFA}, building on a long line of work \cite{cole2015nash,Cole2017nashwelfare,Anari2017MNW} making progress on the problem. The problem was shown to be APX-hard by \citet{lee:j:apx-hardness}, who showed that the problem was inapproximable by a factor of $1.00008$ unless P$=$NP. \citet{garg2018budgetadditive} improved this inapproximability factor to $\sqrt{\frac87} - \varepsilon \approx 1.069$. The best known integrality gap for the configuration LP of this problem is $2^{1/4} \approx 1.189$ \cite{Bei2025MNW}. Our result improves on the existing inapproximability result, moving it closer to the best known integrality gap.

\begin{restatable}{theorem}{thmbudget}\label{thm:budget}
For any constant $\vare > 0$, assuming the unique games conjecture, it is NP-hard to approximate the max budgeted allocation by a factor better than $\frac{243}{227} - \vare \approx 1.07$ when agents have additive valuations.
\end{restatable}

\paragraph{Related Work.} The best known approximation algorithm for this problem achieves an approximation ratio of $\frac43 - c$ for some small $c > 0$ \cite{Kalaitzis2016ImprovedBudgetedAllocation}. This, again, builds on a long line of work \cite{Azar2008BudgetedAllocation,Srinivasan2008Budgeted,Garg2001BudgetedAllocations,Andelman2004BudgetedAllocations,Chakrabarty2010BudgetAdditive} making progress on the problem. The problem was shown to be APX-hard by \citet{Azar2008BudgetedAllocation}, and the inapproximability ratio was improved to $\frac{16}{15}-\vare \approx 1.067$ by \citet{Chakrabarty2010BudgetAdditive}. The best known integrality gap for the configuration LP is $\frac{\sqrt{2} + 1}{2} \approx 1.21$ \cite{Kalaitzis2015ConfigurationLPBudgetedAllocation}. Our result, once again, improves on the existing inapproximability result, moving it closer to the best known integrality gap.

%
%
\begin{restatable}{theorem}{thmgap}\label{thm:gap}
For any constant $\vare > 0$, assuming the unique games conjecture, it is NP-hard to approximate the max generalized assignment problem by a factor better than $\frac{145}{129} - \vare \approx 1.124$ when agents have additive valuations.
\end{restatable}
\paragraph{Related Work} The generalized assignment problem is the most popular problem of the ones we study in this paper, arguably more for its minimization version; see for example, the surveys \cite{Oncan2007GAPSurvey,Morales2004GAPSurvey}. For the maximization problem we study, the best known approximation algorithm achieves an approximation ratio of $\frac{e}{e-1} - c$ for some small $c > 0$ \cite{Feige2006GAPImprovement}. This result improves on the $\frac{e}{e-1}$-approximation of \citet{Fleischer2006GAP}, which improves on the $2$-approximation of \citet{ShmoysTardos1993GAP}\footnote{While the result of \citet{ShmoysTardos1993GAP} is stated for the minimization version of the problem, \citet{Chekuri2000MultipleKnapsack} note that it extends to the maximization version as well.}. \citet{Chekuri2000MultipleKnapsack} showed that the problem is APX-hard, and \citet{Chakrabarty2010BudgetAdditive} improved the factor of inapproximability to $\frac{11}{10} - \vare \approx 1.1$. The best known integrality gap for the configuration LP of the problem is $5/4 = 1.25$ \cite{Feige2006GAPImprovement}.

\begin{remark}
All of our results assume the unique games conjecture (UGC) is true. While this makes our inapproximability results weaker than those without this assumption, we note that the best inapproximability results for several fundamental computer science problems like vertex cover \cite{KhotRegev2008VertexCover}, max cut \cite{Khot2007MaxCutUGC}, and every CSP \cite{Raghavendra2008CSP} assume UGC. 

In the problems we consider, proving lower bounds has been particularly challenging, as evidenced by the lack of strong lower bounds in the literature. The previous best inapproximability results for budgeted allocation and GAP were published over $15$ years ago. In fact, for all the problems we consider, the best known integrality gaps do not match the best approximation ratio for the problem. Given this context, even though we do not prove tight lower bounds, we see our results as a significant contribution towards proving tight inapproximability results for indivisible good allocation problems.
\end{remark}

\begin{table}[t]
\centering
\begin{tabular}{ccccc}
\toprule
\multirow{2}{*}{Problem} & \multicolumn{3}{c}{Prior Work} & \multirow{2}{*}{\makecell{Our Result \\ (UG-Hardness)}} \\
\cmidrule(lr){2-4}
                      & Approximation &  Integrality Gap &  Inapproximability &                 \\
\midrule
Nash Welfare & $e^{1/e}$ \cite{Barman2018FindingFA} & $2^{1/4} \approx 1.189$ \cite{Bei2025MNW} &  $\sqrt{\frac87} \approx 1.069$ \cite{garg2018budgetadditive} & $\sqrt[3]{\frac{81}{65}} \approx 1.076$ \\
\addlinespace
Budgeted Allocation & $\frac43 - \delta_1$ \cite{Kalaitzis2016ImprovedBudgetedAllocation} & $\frac{\sqrt{2} + 1}{2} \approx 1.21$ \cite{Kalaitzis2015ConfigurationLPBudgetedAllocation} &  $\frac{16}{15} \approx 1.067$ \cite{Chakrabarty2010BudgetAdditive} & $\frac{243}{227} \approx 1.07$ \\
\addlinespace
GAP & $\frac{e}{e-1} - \delta_2$ \cite{Feige2006GAPImprovement}& $\frac{5}{4} = 1.25$ \cite{Feige2006GAPImprovement} &  $\frac{11}{10} = 1.1$ \cite{Chakrabarty2010BudgetAdditive} & $\frac{145}{129} \approx 1.124$ \\
\addlinespace
\bottomrule
\end{tabular}
\caption{A summary of our results along with prior work. $\delta_1$ and $\delta_2$ denote small positive constants. The integrality gap refers in all cases to the integrality gap of the configuration LP for the problem.}
\label{tab:results}
\end{table}

\subsection{Overview of Proof Technique}
Our main technical contribution is an instance of the indivisible good allocation problem which can be used as a dictator test. Recall that a boolean function $f:\{0, 1, \dots, q\}^k \rightarrow \{0, 1\}$ is a dictator if the output depends only on the value of one of the $k$ coordinates.

Our dictator test can be converted to NP-hardness results assuming the unique games conjecture using somewhat standard techniques \cite{Bansal2009OneFreeBit,Raghavendra2008CSP,Khot2007MaxCutUGC,Bansal2010HypergraphVC}. To the best of our knowledge, this is the first long code test that uses allocation problems. In this section, we briefly describe the dictator test and why it works. Dictator tests usually require some amount of {\em noise} to ensure soundness. For the sake of readability, we ignore this noise (in this section only) focusing more on the construction and proof ideas.

Our instance is parameterized by a positive integer $R$. We create $3^{R}$ agents. Each agent corresponds to one element of the set $\{0, 1, 2\}^R$. To define the valuation functions, we define a probability distribution $p:\{0, 1, 2\}^{4R} \rightarrow [0,1]$. That is, $p$ is a probability distribution over tuples of $4$ points in $\{0, 1, 2\}^R$. $p$ is defined by the following sampling procedure:
\begin{enumerate}[(i)]
    \item Sample $x^1$ and $x^2$ uniformly at random from $\{0, 1, 2\}^R$.
    \item Define $x^3$ and $x^4$ in $\{0, 1, 2\}^R$ as follows: for each $i \in [R]$, 
    \begin{align*}
        x^3_i \equiv (x^1_i + x^2_i) \mod 3 \qquad \qquad \text{ and }\qquad \qquad  x^4_i \equiv (x^1_i + 2x^2_i) \mod 3.
    \end{align*}
    \item Output $(x^1, x^2, x^3, x^4)$.
\end{enumerate}

For every tuple $(x^1, x^2, x^3, x^4) \in \{0, 1, 2\}^{4R}$ where each $x^1, x^2, x^3, x^4 \in \{0, 1, 2\}^R$, we create one good valued by the agents $x^1, x^2, x^3$ and $x^4$ at value $p(x^1, x^2, x^3, x^4)$ and valued by all the other agents at $0$. We refer to these goods as {\em small} goods; the total welfare achievable from only the small goods is one. We also create $2 \cdot 3^{R-1}$ {\em large} goods that are valued by all the agents at some large value, larger than all of the small goods combined.

The main observation is that we can define a mapping between allocations and boolean functions $f: \{0, 1, 2\}^R \rightarrow \{0, 1\}$. Given any allocation $X$, define a function $f_X: \{0, 1, 2\}^R \rightarrow \{0, 1\}$ such that $f_X(x) = 1$ if and only if the agent corresponding to $x$ receives a large good in the allocation $X$. We use this construction to separate (specific) dictators from functions far away from a dictator. Specifically, we prove a theorem of the following form\footnote{Note that the theorem statement is incorrect since we do not add any noise to the distribution $p$. Adding noise, we get a very similar (and more importantly, correct) theorem statement that can be found in Section \ref{sec:dictator-test}.}. 

\begin{theorem}[Informal]\label{thm:intro-dictator}
In the instance described above, for any allocation $X$ that allocates all the large goods, the following holds:
\begin{description}[font=\normalfont]
    \item[\textsc{Completeness}] If $f_X$ is defined as $f_X(x) = \mathbbm{1}\{x_i > 0\}$ for all $x \in \{0, 1, 2\}^R$ for some $i \in [R]$, then there is an allocation of the small goods such that each agent not allocated a large good in $X$ receives a utility\footnote{The utility of an agent $i$ is their value for their allocated bundle $v_i(X_i)$.} of $3^{-(R-1)}$.
    \item[\textsc{Soundness}] If $f_X$ is `far' from a dictator, then the total utility of all the agents not allocated large goods is at most $\frac{65}{81}$. 
\end{description}
\end{theorem}

Completeness is proved using the observation that for every small good that is positively valued and every $i \in [R]$, there is one agent $z$ who values the good positively and has $z_i = 0$. Therefore, when $f_X(x) = \mathbbm{1}(x_i > 0)$, all small goods can be allocated to agents who do not receive large goods. Moreover, this can be done equitably among the agents $z$ such that $z_i = 0$; there are $3^{R-1}$ such agents. 

Soundness follows from the work of \citet{mossel2010gaussian}. Our distribution $p$ is constructed such that for any $i \in [R]$, the distribution over $(x^1_i, x^2_i, x^3_i, x^4_i)$ is pairwise independent. With this construction, \citet{mossel2010gaussian} states that if a function $f$ is far from a dictator, then it behaves like a random function. We use this to show that of the total value of the small goods (which is one), a $16/81$-fraction of the value must correspond to small goods where all of the agents who value these goods positively are allocated large goods in the allocated $X$.

This test is widely applicable. This is because most reasonable objectives ideally require that the agents who are allocated large goods are not allocated any of the small goods, and that the small goods are distributed equitably among the agents who do not receive any large goods. This is easiest to see with the Nash welfare objective. Theorem \ref{thm:intro-dictator} ensures exactly this: in the positive case, an equitable distribution of small goods is possible while in the negative case, some small goods are inevitably {\em wasted} on the agents who receive large goods. This basic principle is used to prove all the hardness results in this paper. 

Comparing this test to prior work, we note two key differences. Like most dictator tests, our test can also be represented using a constraint satisfaction problem (CSP). In fact, the way we use the test, it is similar to the following dictator test for a function $f: \{0, 1, 2\}^R \rightarrow \{0, 1\}$:
\begin{enumerate}[(i)]
    \item Sample $x^1$ and $x^2$ uniformly at random from $\{0, 1, 2\}^R$.
    \item Define $x^3$ and $x^4$ in $\{0, 1, 2\}^R$ as follows: for each $i \in [R]$, 
    \begin{align*}
        x^3_i \equiv (x^1_i + x^2_i) \mod 3 \qquad \qquad \text{ and }\qquad \qquad  x^4_i \equiv (x^1_i + 2x^2_i) \mod 3.
    \end{align*}
    \item Accept if $f(x^1) f(x^2) f(x^3) f(x^4) = 0$.
\end{enumerate}
This is a dictator test using a CSP with arity $4$ and domain $\{0, 1\}$ where a constraint over four variables $(x^1, x^2, x^3, x^4)$ is satisfied if at least one of them is labeled $0$. Such a CSP is trivial since labeling all variables $0$ satisfies all the constraints. We are able to generate a non-trivial result from this dictator test by imposing a mean constraint on the boolean function $f$; that is, we ensure $f(x) = 1$ on exactly two-thirds of the points $x$. This is not possible to ensure with CSPs but is possible in our setting due to the structure of the allocation problems we consider.

Another difference is that we use the domain $\{0, 1, 2\}^R$ as opposed to the more commonly used $\{0, 1\}^R$. This is a consequence of trying to maximize the inapproximability factor of the problems we consider. A similar dictator test can be defined for all domains $\{0, 1, 
\dots, q\}^R$ such that $q \ge 1$ and $q+1$ is prime. However, the inapproximability factor is maximized for all problems we consider when we use $q=2$. We discuss this in Appendix \ref{sec:generalizations}.
\subsection{Other Related Work}
While we believe our work is the first to prove unique games conjecture based hardness results for allocation problems, similar results exist for scheduling problems \cite{Bansal2009OneFreeBit,Bansal2010HypergraphVC,Svensson2010Precedence}. Scheduling problems generally involve allocating jobs to machines which is similar to our setting of allocating goods to agents. However, these problems usually have another constraint which is exploited to prove hardness results. For example, \cite{Svensson2010Precedence,Bansal2009OneFreeBit} prove hardness results for scheduling problems with precedence constraints, and \cite{Bansal2010HypergraphVC} prove hardness results for shop-type problems where each job needs to be processed on each machine. 

\section{Preliminaries}

\subsection{Basic Notation}
We use $[k]$ to denote the set $\{1, 2, \dots, k\}$, and $[k]_0$ to denote the set $[k] \cup \{0\}$. We use $\mathbbm{1}\{.\}$ to denote the indicator function; this function takes value $1$ if the logical condition within the parentheses is true, and $0$ otherwise. For a variable $x$ and set $D$, we use $\E_{x \in D}[.]$ to denote the expected value when $x$ is distributed uniformly at random over $D$. When $D$ is clear from context, we simply write $\E_x[.]$. Specifically, for a function $f$ over domain $D$, we use $\E_x[f(x)]$ to denote the expectation where $x$ is uniformly distributed over the domain $D$.

\subsection{Indivisible Good Allocation Problem}
An indivisible good allocation instance is defined by a tuple $(N, G, v)$.\footnote{Such tuples also characterize an instance of the fair allocation problem \cite{amanatidis2023fairdivisionsurvey}. However, in our work, most of the objectives we consider have little to do with fairness. So we use the neutral phrase indivisible good allocation problem instead.} 
Here, $N = \{1, 2, \dots, n\}$ refers to a set of $n$ agents, and $G = \{g_1, g_2, \dots, g_m\}$ refers to a set of $m$ goods. Each agent $i \in N$ has a {\em valuation function} $v_i: 2^G \rightarrow \R$; $v_i(S)$ denotes how much value agent $i$ has for the bundle $S$. Throughout this paper, we assume agent valuations are additive; that is, $v_i(S) = \sum_{g \in S} v_i(\{g\})$. 

The goal of this problem is to compute a desirable allocation.
An allocation $X = (X_1, \dots, X_n)$ is an $n$-subpartition of the set of goods. Each agent $i$ is allocated the bundle $X_i$ and receives {\em utility} $v_i(X_i)$. We say a good $g$ is allocated in allocation $X$ if there is some $X_i$ which contains $g$.

\subsection{Allocation Objectives}
We consider the following allocation objectives in this paper.

\noindent \textbf{Max Nash Welfare}: The Nash social welfare of an allocation $X$ is defined as the geometric mean of agent utilities $\NSW(X) = (\prod_{i \in N} v_i(X_i))^{\frac1n}$. An allocation which maximizes the Nash social welfare is referred to as a max Nash welfare allocation.

\noindent \textbf{Max Budgeted Allocation}: Given agent budgets $b_i$ for each $i \in N$, the budgeted social welfare of an allocation $X$ (denoted $\USW_b(X)$) is defined as $\sum_{i \in N} \min \{v_i(X_i), b_i\}$. An allocation which maximizes the budgeted social welfare is referred to as a max budgeted allocation.

\noindent \textbf{Max Generalized Assignment Problem (GAP)}: Given good sizes $s_g$ for each good $g \in G$ and agent capacities $c_i$ for each $i \in N$, an allocation $X$ is said to be {\em feasible} if for each agent $i$, the total size of the bundle $X_i$ is at most $c_i$; that is, $\sum_{g \in X_i} s_g \le c_i$ for each $i \in N$. The utilitarian social welfare (or simply, the social welfare) of an allocation $X$ (denoted $\USW(X)$) is defined as the value $\sum_{i \in N} v_i(X_i)$. An allocation $X$ is said to be a max GAP allocation if it maximizes social welfare subject to feasibility.

In this paper, we mainly deal with the approximability of these objectives. An allocation is said to be an $\alpha$-approximation of some objective ($\alpha \ge 1$) if it achieves an objective value of at least $\frac1{\alpha}\OPT$ where $\OPT$ is the optimal objective value achievable for a specific instance. For example, an allocation $X$ is said to be an $\alpha$-approximation for some instance if $\NSW(X) \ge \frac1{\alpha}\OPT$ where $\OPT$ is the maximum Nash welfare possible in the instance. 
For the max GAP problem, we also require that the allocation $X$ be feasible with respect to the agent capacity constraints in order to be an $\alpha$-approximation.

\subsection{Unique Games Conjecture}

An instance of the unique games problem represented by a tuple $(A \cup B, E, \Pi, [R])$ is defined by a bipartite graph over the node sets $A, B$ and edges $E$. There is a set of labels $[R] = \{1, 2, \dots, R\}$ and a permutation $\pi_{a, b}: [R] \rightarrow [R]$ for each edge $(a, b) \in E$; note that $\pi_{a, b} = \pi^{-1}_{b, a}$. A labeling $\Lambda: A \cup B \rightarrow [R]$ of labels to nodes satisfies an edge $(a, b) \in E$ if $\pi_{a,b}(\Lambda(a)) = \Lambda(b)$. The objective of the problem is to find an assignment that satisfies the most number of edges. 

We use the following version of the unique games conjecture \cite{Bansal2010HypergraphVC,KhotRegev2008VertexCover}, which is equivalent to the original statement \cite{Khot2002UniqueGames}.

\begin{conj}[Unique Games Conjecture]\label{conj:unique-games}
For all $\delta > 0$, there exists a large enough constant $R$ such that given a unique games instance $(A \cup B, E, \Pi, [R])$, it is NP-hard to distinguish between:
\begin{description}[font=\normalfont]
    \item[\textsc{YES Case:}] There is a set $A' \subseteq A$ such that $|A'| \ge (1-\delta)|A|$ and a labeling $\Lambda : A \cup B \rightarrow [R]$ that satisfies every edge $(a, b)$ for $a \in A'$ and $b \in B$.
    \item[\textsc{NO Case:}] No labeling satisfies a $\delta$ fraction of the edges.
\end{description}
Moreover, we can assume the graph is both left and right regular.
\end{conj}

\subsection{Boolean Functions Preliminaries}
In this section, we briefly establish some preliminaries about boolean functions. These preliminaries will allow us to formally define what it means for a boolean function to be far from a dictator, and establish a key tool we can use to deal with such functions. Throughout this section, we assume $q$ is a positive integer. We will study functions of the form $f:[q]_0^R \rightarrow \mathbb{R}$. The input to such a function is a vector $x = (x_1, \dots, x_R)$ of length $R$. We define the $i$-th coordinate of such a vector $x$ as the value of $x_i$. Note that we slightly abuse notation here since $R$ was previously defined as the number of labels in a unique games instance. This is deliberate; in all our results, the length of the input of our boolean functions will be equal to the number of labels of a unique games instance. 

\begin{definition}[Efron-Stein Decomposition \cite{odonnell2014booleanfunctions}]\label{def:efron-stein} 
Every function $f: [q]_0^R \rightarrow \mathbb{R}$ admits a unique decomposition $f = \sum_{S \subseteq [R]} f_S$ such that
\begin{enumerate}[(a)]
\item For every $S$, the function $f_S: [q]_0^{R} \rightarrow \mathbb{R}$ depends only on the coordinates in $S$.
\item For every $S \not \subseteq S'$, and every $y \in [q]_0^R$, we have $\E_{x}[f_S(x) | x_{S'} = y_{S'}] = 0$ where $x_{S'}$ denotes the input vector $x$ restricted to the set of indices $S'$.
\end{enumerate}

For any positive integer $d$, we define $f^{\le d}$ as the function $\sum_{S: |S| \le d} f_S$ where $f_S$ is given by the Efron-Stein decomposition.
\end{definition}

\begin{definition}[Influence \cite{odonnell2014booleanfunctions}]
For a function $f:[q]_{0}^R \rightarrow \mathbb{R}$ and $i \in [R]$, we define the influence of the $i$-th coordinate (denoted $\Inf_i(f)$) as 
\begin{align*}
	\Inf_i(f) = \sum_{S: i \in S} \E_x[f_S(x)^2],
\end{align*}
where $f_S$ is given by the Efron-Stein decomposition.
We also define the degree $d$-influence of $f$ (denoted $\Inf_i^{\le d}(f)$) as the value $\Inf_i(f^{\le d})$. Note that $\Inf_i^{\le d}(f) = \sum_{S: i \in S, |S|\le d} \E_x[f_S(x)^2]$.
\end{definition}

When $f$ is a dictator, it depends only on one coordinate and therefore the influence values of the function $\Inf_i(f)$ will be very high for one coordinate and $0$ for all the other coordinates. Naturally, we can define a function as being far from a dictator if each coordinate has low influence. In this paper, we will say a function is far from a dictator if for some $d$ and $\tau$ (that are proof dependent), $\Inf_i^{\le d}(f) \le \tau$ for all $i \in [R]$. We will need the following useful lemma, the proof for which can be found in \citet[Lemma 2.8]{Bansal2010HypergraphVC}.

\begin{lemma}\label{lem:d-inf}
For any function $f:[q]_0^R \rightarrow [0, 1]$, the number of coordinates $i$ such that $\Inf_i^{\le d}(f) \ge \tau$ is at most $d/\tau$ for any $d \ge 1$ and $\tau > 0$. 
\end{lemma}

\begin{definition}[Balanced and Pairwise Independent\cite{Austrin2009PairwiseIndependence}]
For some positive integer $k$, let $([q]_0^k, \eta)$ be a probability space. We say that $\eta$ is balanced if for every $i \in [k]$, $j \in [q]_0$, we have $\Pr_{w \in ([q]^k_0, \eta)}[w_i = j] = \frac1{q+1}$. 

We say that $\eta$ is pairwise independent if for every $i_1, i_2 \in [k]$ ($i_1 \ne i_2$), and $j_1, j_2 \in [q]_0$, we have $\Pr_{w \in ([q]^k_0, \eta)}[w_{i_1} = j_1, w_{i_2} = j_2] = \frac1{(q+1)^2}$.
\end{definition}

The correctness of our reductions will heavily use the following result by \citet{mossel2010gaussian}.

\begin{theorem}(\citet[Theorem 6.6]{mossel2010gaussian})\label{thm:mossel-pairwise}
For any positive integers $q$ and $k$, let $([q]^k_0, \eta)$ be a probability space such that 
\begin{enumerate}[(a)]
\item $\eta$ is pairwise independent and balanced, and
\item For all $w \in [q]^k_0$, $\eta(w) \ge \alpha > 0$ for some constant $\alpha$.
\end{enumerate}
Then, for all $\vare > 0$, there exists a $\tau > 0$ and $d > 0$ such that the following holds for all positive integers $R$. Let $f_1, f_2, \dots, f_k$ be functions such that $f_j: [q]^R_0 \rightarrow [0, 1]$ and for all $i \in [R]$, $\Inf_i^{\le d}(f_j) \le \tau$, then
\begin{align*}
\left | \E_{w^1, \dots, w^R} \left [\prod_{j = 1}^k f_j(w^1_{j}, \dots, w^R_{j}) \right ] - \prod_{j = 1}^k\left (\E_{w^1, \dots, w^R}\left [f_j(w^1_{j}, \dots, w^R_{j}) \right ] \right ) \right | \le \vare.
\end{align*}
Here, each $w^i$ is independent and distributed according to the probability space $([q]^k_0, \eta)$.
\end{theorem}

The statement of \citet[Theorem 6.6]{mossel2010gaussian} is much more general than what is stated here. The version we use is very similar to the one used in \citet[Theorem 2.7]{Austrin2009PairwiseIndependence}.

The result essentially states that low influence functions (functions far away from dictators) behave like random functions when each $w^j$ is balanced and pairwise independent. 

\section{Dictator Test}\label{sec:dictator-test}
In this section, we describe our dictator test. Throughout this section, let $R$ be any positive integer and let $\vare \in (0, 1)$ be a small rational constant.
As described in the introduction, our dictator test uses a probability distribution to determine the value of each good in our instance. 

To define this distribution, we define a probability distribution $\eta:\{0, 1, 2\}^4 \rightarrow [0, 1]$ as follows: $\eta$ is the uniform distribution over all tuples of the form $(a, b, a + b, a + 2b)$ where $a, b \in \{0, 1, 2\}$ and addition is modulo $3$. There are nine such tuples, one for each $a$ and $b$. Formally, $\eta$ is defined as follows (addition again, is modulo $3$):
\begin{align*}
	\eta(a, b, c, d) = 
	\begin{cases}
		\frac19 & c = a + b \text{ and } d = a + 2b \\
		0 & \text{otherwise}
	\end{cases}
\end{align*}

\begin{claim}\label{claim:eta-properties}
The distribution $\eta$ satisfies the following properties:
\begin{enumerate}[(i)]
\item $\eta$ is balanced and pairwise independent
\item For $(a, b, c, d) \in \{0, 1, 2\}^4$ sampled according to the distribution $\eta$, at least one of $a, b, c$ and $d$ is $0$ with probability $1$. 
\end{enumerate}
\end{claim}
\begin{proof}
(i) can be verified in many ways, the simplest being checking the nine tuples in the support of the distribution $\eta$. For a more elegant proof, we refer the reader to \citet[Lemma 4.2]{Austrin2009PairwiseIndependence}. 

To prove (ii), we use the fact that $\eta$ is supported by tuples of the form $(a, b, a+b, a+2b)$ where $a, b \in \{0, 1, 2\}$ and addition is defined modulo $3$. If $b = 0$, we are done. If $b \ne 0$, then $\{b, 2b\}$ is equivalent to the set $\{1, 2\}$ modulo $3$. Therefore, at least one of $a$, $a + b$ and $a + 2b$ must be $0$. 
\end{proof}

While $\eta$ satisfies several nice properties, it does not assign a positive probability to every element in $\{0, 1, 2\}^4$. We therefore cannot apply Theorem \ref{thm:mossel-pairwise} with this distribution. We resolve this issue by adding a small amount of noise to $\eta$. Formally, we define another probability distribution $\eta'_{\vare}: \{0, 1, 2\}^4 \rightarrow [0, 1]$ as follows: for every $w \in \{0, 1, 2\}^4$, $\eta'_{\vare}(w) = (1-\vare)\eta(w) + \vare \mu(w)$ where $\mu$ is the uniform distribution over $\{0, 1, 2\}^4$. 

\begin{claim}\label{claim:eta-prime-properties}
The distribution $\eta'_{\vare}$ satisfies the following properties:
\begin{enumerate}[(i)]
\item $\eta'_{\vare}$ is balanced and pairwise independent.
\item For every $w \in \{0, 1, 2\}^4$, $\eta'_{\vare}(w) \ge \frac{\vare}{81}$, a constant.
\item For $(a, b, c, d) \in \{0, 1, 2\}^4$ sampled according to the distribution $\eta'_{\vare}$, at least one of $a, b, c$ and $d$ is $0$ with probability at least $1 - \vare$. 
\end{enumerate}
\end{claim}
\begin{proof}
(i) follows from $\eta$ and $\mu$ being balanced pairwise independent distributions (Claim \ref{claim:eta-properties}(i)).

(ii) follows from the noise we add to create $\eta'_{\vare}$. For any $w \in \{0, 1, 2\}^4$, $\eta'_{\vare}(w) \ge \vare \mu(w) = \frac{\vare}{81}$.

(iii) follows from Claim \ref{claim:eta-properties}(ii) and the definition of $\eta'_{\vare} = (1 - \vare)\eta + \vare \mu$.
\end{proof}

We use $\eta'_{\vare}$ to define the probability distribution $p_{\vare}$ we use to determine agent valuations. 

\begin{definition}\label{def:p-distribution}
For any $\vare \in (0, 1)$, the probability distribution $p_{\vare}: \{0, 1, 2\}^{4R} \rightarrow [0, 1]$ is defined over tuples of four elements in $\{0, 1, 2\}^R$. It is defined as follows: $p_{\vare}(x^1, x^2, x^3, x^4) = \prod_{i \in [R]}\eta'_{\vare}(x^1_i, x^2_i, x^3_i, x^4_i)$. In simple words, the $i$-th coordinate of each vector $x^1, x^2, x^3$ and $x^4$ is distributed according to the distribution $\eta'_{\vare}$.
\end{definition}

We are now ready to define our dictator test instance. 
\begin{prob}\label{prob:dictator-test}
We have a set of $3^R$ agents, one for every element in $\{0, 1, 2\}^R$. For every tuple of four (not necessarily distinct) agents $(x^1, x^2, x^3, x^4)$ we define one good which is valued at $p_{\vare}(x^1, x^2, x^3, x^4)$ by the agents $x^1, x^2, x^3$ and $x^4$, and valued at $0$ by all the other agents. This creates $(3^R)^4$ goods. We refer to these goods as {\em small} goods. We define a set of $2 \cdot 3^{R-1}$ large goods, each of which is valued by all the agents at $\frac1{\vare}$.   
\end{prob}

For every allocation $X$ in this instance, we define the boolean function $f_X: \{0, 1, 2\}^R \rightarrow \{0, 1\}$ as $f_X(x) = 1$ if and only if $x$ receives a large good in the allocation $X$.

\begin{theorem}[Dictator Test]\label{thm:dictator-test}
In an instance from Problem Instance Family \ref{prob:dictator-test}, let $X$ be an allocation where $2 \cdot 3^{R-1}$ agents receive a large good, and let $N'$ be the set of agents which do not receive any large goods. Then:
\begin{description}[font=\normalfont]
    \item[\textsc{Completeness.}] If $f_X$ is defined as $f_X(x) = \mathbbm{1}\{x_i > 0\}$ for some $i \in [R]$, then there is an allocation $X'$ with $f_X = f_{X'}$ such that for all agents in $u \in N'$, $v_u(X'_u) \ge \frac{1}{3^{R-1}}(1 - \vare)$.
    \item[\textsc{Soundness.}] There exists a $d > 0$ and $\tau > 0$ such that if $f_X$ has $\Inf^{\le d}_i(f_X) \le \tau$ for all $i \in [R]$, then $\sum_{u \in N'} v_u(X_u) \le \frac{65}{81} + \vare$. 
\end{description}
\end{theorem}
\begin{proof}
We start with the completeness. Fix some $i \in [R]$ and let $X$ allocate a large good to all agents $x$ with $x_i > 0$. We show that there is an allocation of the small goods such that each agent $x$ with $x_i = 0$ receives a utility of at least $\frac{1- \vare}{3^{R - 1}}$. We define our allocation using a function $\chi:\{0, 1, 2\}^4 \rightarrow [4]_{0}$ where $\chi(y)$ is equal to the first coordinate of $(y_1, y_2, y_3, y_4)$ which is equal to $0$; if no such coordinate exists, $\chi(y) = 0$. For each small good $g$ defined by a tuple $(x^1, x^2, x^3, x^4)$, we allocate the good to $x^t$ where $t = \chi(x^1_i, x^2_i, x^3_i, x^4_i)$ if it is non-zero and $t = 1$ otherwise. Denote the resulting allocation using $X'$. We will need the following observation about $\chi$, which follows immediately from Claim \ref{claim:eta-prime-properties}(iii).

\begin{obs}\label{obs:chi}
$\E_{y \sim \eta'_{\vare}}[\mathbbm{1}\{\chi(y) > 0\}] \ge 1 - \vare$.
\end{obs} 

At a high level, completeness follows from the fact that goods are allocated solely based on the $i$-th coordinate. Therefore, if a good $g$ is allocated to $x^j$ solely based on $x^j_i$, all the other coordinates of $x^j$ are distributed uniformly at random. So, each agent $z$ such that $z_i = 0$ receives a $\frac{1}{3^{R-1}}$ share of the total utility. This combined with the fact at least a $(1- \vare)$-fraction of the total utility gets allocated to agents who are not allocated a large good (Claim \ref{claim:eta-prime-properties}(iii)), we get that each agent $z$ such that $z_i = 0$ receives a utility of at least $\frac{1 - \vare}{3^{R-1}}$. 

Formally, for any agent $z$ such that $z_i = 0$, we can write out its utility to lower bound it. In this sequence of inequalities, we use the notation $x_S$ which for any $S \subseteq [R]$, restricts the vector $x$ to its coordinates in $S$.  
We start with a simple lower bound on $v_z(X'_z)$.
\begin{align*}
v_z(X'_z) &\ge \sum_{t = 1}^4 \sum_{(x^1, x^2, x^3, x^4) \in \{0, 1, 2\}^{4R}} p_{\vare}(x^1, x^2, x^3, x^4) \mathbbm{1}\{\chi(x^1_i, x^2_i, x^3_i, x^4_i) = t\} \mathbbm{1}\{x^t = z\}.
\end{align*}

The right hand side of the above inequality iterates over $t \in [4]$ and all the goods checking using indicator functions if 
\begin{inparaenum}[(a)]
\item each good defined by $(x^1, x^2, x^3, x^4)$ is allocated to $x^t$, and 
\item $x^t = z$.
\end{inparaenum}
It is easy to see that this lower bounds $v_z(X'_{z})$. To simplify our calculations, we can express this lower bound as an expectation. 

\begin{align*}
v_z(X'_z) &\ge \sum_{t = 1}^4 \E_{(x^1, x^2, x^3, x^4) \sim p_{\vare}} \left [\mathbbm{1}\{\chi(x^1_i, x^2_i, x^3_i, x^4_i) = t\}  \mathbbm{1}\{x^t = z\} \right ]. 
\end{align*}

Next, we use the fact that $z_i = 0$ to simplify the indicator functions. 
\begin{align*}
v_z(X'_z) &\ge \sum_{t = 1}^4 \E_{(x^1, x^2, x^3, x^4) \sim p_{\vare}} \left [\mathbbm{1}\{\chi(x^1_i, x^2_i, x^3_i, x^4_i) = t\} \mathbbm{1}\{x^t_i = 0\} \mathbbm{1}\{x^t_{[R] \setminus \{i\}} = z_{[R] \setminus \{i\}}\}\right ] \\
&= \sum_{t = 1}^4 \E_{(x^1, x^2, x^3, x^4) \sim p_{\vare}} \left [\mathbbm{1}\{\chi(x^1_i, x^2_i, x^3_i, x^4_i) = t\} \mathbbm{1}\{x^t_{[R] \setminus \{i\}} = z_{[R] \setminus \{i\}}\}\right ].
\end{align*}
The equality follows from the fact that $x^t_i = 0$ is implied by $\chi(x^1_i, x^2_i, x^3_i, x^4_i) = t$ for any $t > 0$. We simplify this further using the fact that the $i$-th coordinate is distributed independently from the other coordinates. Therefore, the events $\chi(x^1_i, x^2_i, x^3_i, x^4_i) = t$ and $x^t_{[R] \setminus \{i\}} = z_{[R] \setminus \{i\}}$ are independent. This gives us

\begin{align*}
v_z(X'_z) &\ge \sum_{t = 1}^4 \E_{(x^1, x^2, x^3, x^4) \sim p_{\vare}} \left [\mathbbm{1}\{\chi(x^1_i, x^2_i, x^3_i, x^4_i) = t\} \right ] \E_{(x^1, x^2, x^3, x^4) \sim p_{\vare}} \left [ \mathbbm{1}\{x^t_{[R] \setminus \{i\}} = z_{[R] \setminus \{i\}}\}\right ] \\
&= \sum_{t = 1}^4 \frac{1}{3^{R-1}} \E_{(x^1, x^2, x^3, x^4) \sim p_{\vare}} \left [\mathbbm{1}\{\chi(x^1_i, x^2_i, x^3_i, x^4_i) = t\} \right ] \\
&= \frac{1}{3^{R-1}} \E_{(x^1_i, x^2_i, x^3_i, x^4_i) \sim \eta'_{\vare}} \left [\mathbbm{1}\{\chi(x^1_i, x^2_i, x^3_i, x^4_i) > 0\} \right ]\\
&\ge \frac{1 - \vare}{3^{R-1}}.
\end{align*}

The first equality follows from the fact that the distribution $\eta'_{\vare}$ is balanced, and the final inequality follows from Observation \ref{obs:chi}.
Since $z$ was picked arbitrarily among the agents who do not receive a large good, this completes the completeness proof. 

We now prove soundness. Let $d$ and $\tau$ denote the $d$ and $\tau$ obtained from applying Theorem \ref{thm:mossel-pairwise} with distribution $\eta'_{\vare}$ and constant $\vare$. This is valid since $\eta'_{\vare}$ satisfies the requirements of Theorem \ref{thm:mossel-pairwise} (Claim \ref{claim:eta-prime-properties}). Our analysis uses the following key observation: a small good $g$ defined by the tuple $(x^1, x^2, x^3, x^4)$ is valued positively only by agents who receive a large good if and only if $f_X(x^1)f_X(x^2)f_X(x^3)f_X(x^4) = 1$. Therefore, we can upper bound the total utility of agents who do not receive large goods as follows:

\begin{align}
\sum_{u \in N'} v_u(X_u) &\le \sum_{(x^1, x^2, x^3, x^4) \in \{0, 1, 2\}^{4R}} p_{\vare}(x^1, x^2, x^3, x^4) \left [1 -  \prod_{j = 1}^4 f_X(x^j) \right ] \notag \\
&= 1 - \E_{(x^1, x^2, x^3, x^4) \sim p_{\vare}} \left [\prod_{j = 1}^4 f_X(x^j) \right ] \label{eq:dictator-1}
\end{align}

We invoke Theorem \ref{thm:mossel-pairwise} with the low bounded-degree influence function $f_X$. This gives us

\begin{align}
\left | \E_{(x^1, x^2, x^3, x^4) \sim p_{\vare}} \left [\prod_{j = 1}^4 f_X(x^j) \right ]  - \prod_{j = 1}^4 \E_{(x^1, x^2, x^3, x^4) \sim p_{\vare}}[f_X(x^j)] \right | \le \vare \label{eq:mossel-dictator}
\end{align}

Plugging \eqref{eq:mossel-dictator} into \eqref{eq:dictator-1}, we get
\begin{align*}
\sum_{u \in N'} v_u(X_u) &\le 1 - \prod_{j = 1}^4 \E_{(x^1, x^2, x^3, x^4) \sim p_{\vare}}[f_X(x^j)] + \vare \\
&= 1 - \left (\frac23 \right )^4 + \vare
= \frac{65}{81} + \vare.
\end{align*}

In the first equality, we use the fact that $\eta'_{\vare}$ is balanced; therefore $\E_{(x^1, x^2, x^3, x^4) \sim p}[f_X(x^j)]$ denotes the expected value of $f_X$ under the uniform distribution over the set $\{0, 1, 2\}^R$. The value of $\E_x[f_X(x)]$ is $2/3$ since the allocation $X$ allocates a large good to $2 \cdot 3^{R - 1}$ agents. This completes the soundness proof. 
\end{proof}

\begin{remark}
Our idea to use the distribution $\eta$ comes from \citet[Theorem 1.3]{Austrin2009PairwiseIndependence}, who construct a family of pairwise independent distributions with one of them being $\eta$. They use pairwise independent distributions to prove hardness results for CSPs. We apply this idea to allocation problems. 
\end{remark}

\section{Meta Theorem}\label{sec:meta}
In this section, we use the dictator test from the previous section to prove a hardness result. This result will serve as a meta theorem that we use to prove inapproximability results for Nash welfare and budgeted allocation. Our meta theorem will use instances of the following form, parameterized by a small rational constant $\vare \in (0, 1)$:

\begin{prob}\label{prob:meta}
We have a set of $n$ agents $N$ that can be partitioned into a set of $k$ equally sized groups $N_1, \dots, N_k$. For each group $k' \in [k]$, we have a set of $\frac{2|N_{k'}|}{3}$ large goods each valued at $\frac{1}{\vare}$ by the agents in $N_{k'}$, and $0$ by all the other agents. For all the other goods (which are not large goods) $G'$, it holds that for any agent $u \in N$: $v_u(G')\le (1 + \vare)$.
\end{prob}

\begin{theorem}[Meta Theorem]\label{thm:meta}
For every rational constant $\vare \in (0, 1)$, assuming the unique games conjecture is true, it is NP-hard to distinguish between the following two cases given an instance from Problem Instance Family \ref{prob:meta} (parameterized by $\vare$):
\begin{description}[font=\normalfont]
    \item[\textsc{YES Case:}] There is an allocation $X$ where $2n/3$ agents are allocated a large good they value at $\frac1{\vare}$ and all the agents who are not allocated a large good receive a utility of at least $\frac3n (1 - \vare)$.
    \item[\textsc{NO Case:}] For every allocation $X$ where $2n/3$ agents receive a large good they value at $\frac1{\vare}$, the total utility of agents who do {\em not} receive a large good is at most $\frac{65}{81} + 3\vare$. That is, $\sum_{u \in N'}v_u(X_u) \le \frac{65}{81} + 3\vare$ where $N'$ is the set of agents who do not receive a large good in $X$.
\end{description} 
\end{theorem}

\subsection{Proof of Theorem \ref{thm:meta}}
We use a reduction from the unique games problem. Let $\mathcal{L} = (A \cup B, E, \Pi, [R])$ be an instance from Conjecture \ref{conj:unique-games} with $\delta = \frac{\vare \tau^2}{8d}$, where $\tau$ and $d$ are obtained from Theorem \ref{thm:mossel-pairwise} with distribution $\eta'_{\vare}$ and constant $\vare$. Before we present our construction, we first establish some notation. For any node $a \in A \cup B$, we use $\Nbd(a)$ to denote the set of neighbors of $a$ in $\cal L$. Recall that the graph is both left and right regular; we use $\delta_A$ and $\delta_B$ to denote the degrees of nodes in $A$ and $B$ respectively. For a permutation $\pi: [R] \rightarrow [R]$ and a vector $x \in \{0, 1, 2\}^R$, we define $x \circ \pi = (x_{\pi(1)}, x_{\pi(2)}, \dots, x_{\pi(R)})$. Similarly, given a set $S \subseteq [R]$, we define $\pi(S) := \cup_{i \in S}\{\pi(i)\}$. 

\paragraph{Instance.}
We construct our indivisible good allocation instance as follows. We have a set of $|A|3^R$ agents; specifically, we have an agent for each tuple of the form $(a, x)$ where $a \in A$ and $x \in \{0, 1, 2\}^R$. For every tuple of the form $(b, (a^1, x^1), (a^2, x^2), (a^3, x^3), (a^4, x^4))$ where $b \in B$, $a^1, a^2, a^3, a^4 \in \Nbd(b)$, and $x^1, x^2, x^3, x^4 \in \{0, 1, 2\}^R$, we create one good. We refer to these goods as {\em small} goods. The value of these goods is defined by a probability distribution $q$ with 
\begin{align}
q(b, (a^1, x^1), (a^2, x^2), (a^3, x^3), (a^4, x^4)) = \frac{1}{|B| \delta_B^4} p_{\vare}(x^1 \circ \pi_{b, a^1}, x^2 \circ \pi_{b, a^2}, x^3 \circ \pi_{b, a^3}, x^4 \circ \pi_{b, a^4}), \notag
\end{align}
where $p_{\vare}$ is defined in Definition \ref{def:p-distribution}. We can define $q$ equivalently by the following sampling procedure:
\begin{enumerate}[(i)]
\item Sample $b$ uniformly at random from $B$.
\item Sample four neighbors of $b$, $(a^1, a^2, a^3, a^4)$ uniformly at random and with replacement. 
\item Sample $(x^1, x^2, x^3, x^4)$ from the distribution $p_{\vare}$.
\item Output $(b, (a^1, x^1 \circ \pi_{a^1, b}), (a^2, x^2 \circ \pi_{a^2, b}), (a^3, x^3 \circ \pi_{a^3, b}), (a^4, x^4 \circ \pi_{a^4, b}))$.
\end{enumerate}

This uses the fact that $\pi_{a, b} = \pi^{-1}_{b, a}$ for any edge $(a, b) \in E$. The good $g$ defined by the tuple \\ $(b, (a^1, x^1), (a^2, x^2), (a^3, x^3), (a^4, x^4))$ is valued by the four agents $(a^1, x^1), (a^2, x^2), (a^3, x^3), (a^4, x^4)$ at value $q(b, (a^1, x^1), (a^2, x^2), (a^3, x^3), (a^4, x^4))$ and $0$ by all the other agents. For every $a \in A$, we create $2 \cdot 3^{R-1}$ {\em large} goods which are valued by all the agents of the form $(a, x)$ with $x \in \{0, 1, 2\}^R$ at $\frac1{\vare}$ and all the other agents at $0$. We also create $\lfloor \delta |A| \rfloor 3^{R - 1}$ dummy goods valued by all the agents at $\frac{1 - \vare}{|A|3^{R-1}}$. This instance can be easily constructed in polynomial time since $R$ is a constant. Note specifically, that since $R$ is a constant and $\vare$ is a rational constant, the value of $p_{\vare}(.)$ is always a rational constant. This implies that the values of all the goods are rational values with a polynomially bounded numerator and denominator.

We first show that this instance belongs to Problem Instance Family \ref{prob:meta}. The set of agents we construct can be partitioned using the vertex $a$. That is, for each $a \in A$ we construct a group of agents consisting of agents of the form $(a, x)$ with $x \in \{0, 1, 2\}^R$. This creates a partition of agents into $|A|$ equally sized groups $N_1, \dots, N_{|A|}$. For any group $N_{k'}$ our construction has a set of $\frac{2|N_{k'}|}{3}$ large goods that are valued by each agent in $N_{k'}$ at $\frac1{\vare}$ and $0$ by all other agents. The non-large goods in our instance consist of small goods and dummy goods. The small goods have value at most $1$ for any agent since their value is governed by the probability distribution $q$. The dummy goods have value at most 
\begin{align}
\lfloor \delta |A| \rfloor 3^{R - 1} \times \frac{1 - \vare}{|A|3^{R-1}} \le \delta \le \vare. \label{eq:dummy}
\end{align}
This shows that the non-large goods have value at most $1 + \vare$ for any agent. Therefore, our constructed instance belongs to Problem Instance Family \ref{prob:meta}.
The remainder of this proof deals with the hardness result.

\paragraph{YES Case.} In the YES Case of the unique games instance, there is a set $A' \subseteq A$ such that $|A'| \ge (1- \delta)|A|$ and a labeling $\Lambda$ that satisfies all the edges containing some node in $A'$. Using this labeling, we define an allocation $X$ for our constructed indivisible good allocation instance. For every $a \in A \setminus A'$, we allocate the large goods arbitrarily such that $2 \cdot 3^{R-1}$ agents of the form $(a, x)$ receive a large good they value at $\frac1{\vare}$. To the agents $(a, x)$ who do not receive a large good, we allocate a dummy good. This is possible to do since $|A \setminus A'| \le \delta |A|$ from the assumption of the YES case. This ensures that all agents of the form $(a, x)$ where $a \in A \setminus A'$ either receive a large good they value at $\frac1{\vare}$ or receive a dummy good they value at $\frac{1 - \vare}{|A|3^{R-1}}$. 

For every $a \in A'$, we allocate a large good to $(a, x)$ if $x_{\Lambda(a)} > 0$. To allocate the small goods, we use the function $\chi:\{0, 1, 2\}^{4} \rightarrow [4]_0$ defined in the previous section. Recall that $\chi(w)$ is equal to the first coordinate $j$ such that $w_j = 0$ if such a coordinate exists, and $0$ otherwise. 

Consider a good $g$ defined by the tuple $(b, (a^1, x^1 \circ \pi_{a^1, b}), \dots (a^4, x^4 \circ \pi_{a^4, b}))$. The probability of this tuple according to $q$ as given by the sampling procedure is $\frac{p_{\vare}(x^1, x^2, x^3, x^4)}{|B|\delta_B^4}$. Given the good $g$, we allocate the good to the agent $(a^t, x^t \circ \pi_{a^t, b})$ where $t = \chi \left (x^1_{\Lambda(b)}, x^2_{\Lambda(b)}, x^3_{\Lambda(b)}, x^4_{\Lambda(b)} \right )$ if $\chi \left (x^1_{\Lambda(b)}, x^2_{\Lambda(b)}, x^3_{\Lambda(b)}, x^4_{\Lambda(b)} \right ) > 0$ and $t = 1$ otherwise. We will deal with goods using this more complicated tuple notation since it simplifies the utility calculations. Specifically, we will exploit the fact that the tuple $(x^1, x^2, x^3, x^4)$ is distributed according to $p_{\vare}$, and $a^1, a^2, a^3, a^4$ are randomly sampled neighbors of $b$.

Let $(a, z)$ be an agent such that $a \in A'$ and $z_{\Lambda(a)} = 0$. We need to show that $v_{(a, z)}(X_{(a, z)}) \ge \frac{1-\vare}{|A|3^{R-1}}$. We can lower bound the utility of $(a, z)$ by writing it as an expectation, similar to the previous section. To simplify our calculations, we expand $\Pi$ to include all pairs of nodes: we define $\pi_{a, b}$ to be any permutation with $\pi_{a, b}(\Lambda(a)) = \Lambda(b)$ when $a$ and $b$ are not adjacent in $\cal L$. For every small good $g$, we associate it with a tuple of the form $(b, (a^1, x^1 \circ \pi_{a^1, b}), \dots (a^4, x^4 \circ \pi_{a^4, b}))$. Abusing notation slightly, we write $q(g)$ to mean $q(b, (a^1, x^1 \circ \pi_{a^1, b}), \dots (a^4, x^4 \circ \pi_{a^4, b}))$.
We have the following lowerbound on $v_{(a, z)}(X_{(a, z)})$.
\begin{align*}
v_{(a, z)}(X_{(a, z)}) &\ge \sum_{t = 1}^4 \sum_g q(g) \mathbbm{1}\left \{\chi \left (x^1_{\Lambda(b)}, \dots, x^4_{\Lambda(b)} \right ) = t \right \} \mathbbm{1}\{a^t = a\} \mathbbm{1}\{x^t \circ \pi_{a^t, b} = z\}\\
&= \sum_{t = 1}^4 \E_{g \sim q} \left [ \mathbbm{1}\left \{\chi \left (x^1_{\Lambda(b)}, \dots, x^4_{\Lambda(b)} \right ) = t \right \} \mathbbm{1}\{a^t = a\} \mathbbm{1}\{x^t \circ \pi_{a^t, b} = z\} \right ] \\
&= \sum_{t = 1}^4 \E_{g \sim q} \left [ \mathbbm{1}\left \{\chi \left (x^1_{\Lambda(b)}, \dots,  x^4_{\Lambda(b)} \right ) = t \right \} \mathbbm{1}\{a^t = a\} \mathbbm{1}\{x^t \circ \pi_{a, b} = z\} \right ]. 
\end{align*}
In the final equality, we used the fact that $\mathbbm{1}\{a^t = a\} \mathbbm{1}\{x^t \circ \pi_{a^t, b} = z\} = \mathbbm{1}\{a^t = a\} \mathbbm{1}\{x^t \circ \pi_{a, b} = z\}$ since it takes the value $0$ if $a^t \ne a$. 

We know that $z_{\Lambda(a)} = 0$. Therefore, the above expression (inside the expectation) takes value $1$ only when $(x^t \circ \pi_{a, b})_{\Lambda(a)} = 0$. This is equivalent to $x^t_{\pi_{a, b}(\Lambda(a))} = 0$. Since $a \in A'$, the labeling $\Lambda$ satisfies $\pi_{a, b}(\Lambda(a)) = \Lambda(b)$ for all $b \in B$ (via our expansion of $\Pi$). Therefore, $(x^t \circ \pi_{a, b})_{\Lambda(a)} = 0$ is equivalent to $x^t_{\Lambda(b)} = 0$ which is implied by  $\mathbbm{1}\left \{\chi \left (x^1_{\Lambda(b)}, x^2_{\Lambda(b)}, x^3_{\Lambda(b)}, x^4_{\Lambda(b)} \right ) = t \right \}$ for any positive $t$. So we can exclude it from our expectation calculations. This simplifies the utility of $v_{(a, z)}(X_{(a, z)})$ to the following:
\begin{align*}
v_{(a, z)}(X_{(a, z)}) &\ge \sum_{t = 1}^4 \E_{g \sim q} \left [ \mathbbm{1}\left \{\chi \left (x^1_{\Lambda(b)}, \dots, x^4_{\Lambda(b)} \right ) = t \right \} \mathbbm{1}\{a^t = a\} \mathbbm{1}\left \{(x^t \circ \pi_{a, b})_{R \setminus \{\Lambda(a)\}} = z_{R \setminus \{\Lambda(a)\}} \right \} \right ].
\end{align*}

To simplify this further, we use the fact that, conditioned on a specific value of $b$, the events \\ $\left \{\chi \left (x^1_{\Lambda(b)}, x^2_{\Lambda(b)}, x^3_{\Lambda(b)}, x^4_{\Lambda(b)} \right ) = t \right \}$, $\{a^t = a\}$ and $\left \{(x^t \circ \pi_{a, b})_{R \setminus \{\Lambda(a)\}} = z_{R \setminus \{\Lambda(a)\}} \right \} $ are independent. This is easy to see when the distribution $q$ is understood using the sampling procedure. Each coordinate of $x^t$ is independent of the other coordinates, and $(x^t \circ \pi_{a, b})_{R \setminus \{\Lambda(a)\}}$ excludes $x^t_{\Lambda(b)}$. Moreover, the choice of $a^t$ is independent of the $(x^1, x^2, x^3, x^4)$ given $b$. Since $b$ is distributed uniformly over the set $B$, this simplifies our utility expression to:

\begin{align*}
v_{(a, z)}(X_{(a, z)}) &\ge \sum_{t = 1}^4 \sum_{b' \in B} \frac1{|B|} \left [\E_{g \sim q} \left [ \mathbbm{1}\left \{\chi \left (x^1_{\Lambda(b)}, \dots, x^4_{\Lambda(b)} \right ) = t \right \} \,\middle |\, b = b' \right ]  \E_{g \sim q} \left [ \mathbbm{1}\left \{a^t = a \right \} \,\middle |\, b = b' \right ] \right . \\
& \qquad \qquad \qquad \left . \E_{g \sim q} \left [ \mathbbm{1}\left \{(x^t \circ \pi_{a, b})_{R \setminus \{\Lambda(a)\}} = z_{R \setminus \{\Lambda(a)\}} \right \} \, \middle |\, b = b' \right ] \right ].
\end{align*} 

We can analyze each term separately to simplify it. Since $a^t$ is chosen uniformly at random from the neighbors of $b$,
\begin{align}
\E_{g \sim q} \left [ \mathbbm{1}\left \{a^t = a \right \} \,\middle |\, b = b' \right ] = \frac1{\delta_B}\mathbbm{1}\{b' \in \Nbd(a)\}. \label{eq:1}
\end{align}
Since the value of $x^t_i$ is sampled from the balanced distribution $\eta'_{\vare}$,
\begin{align}
\E_{g \sim q} \left [ \mathbbm{1}\left \{(x^t \circ \pi_{a, b})_{R \setminus \{\Lambda(a)\}} = z_{R \setminus \{\Lambda(a)\}} \right \} \, \middle |\, b = b' \right ] = \frac1{3^{R-1}}. \label{eq:2}
\end{align}
Using \eqref{eq:1} and \eqref{eq:2}, we can simplify the lower bound on $v_{(a, z)}(X_{(a, z)})$ even further.
\begin{align*}
v_{(a, z)}(X_{(a, z)}) &\ge \sum_{t = 1}^4 \sum_{b' \in B} \frac1{|B|} \left [\E_{g \sim q} \left [ \mathbbm{1}\left \{\chi \left (x^1_{\Lambda(b)}, \dots, x^4_{\Lambda(b)} \right ) = t \right \} \,\middle |\, b = b' \right ]  \frac{\mathbbm{1}\{b' \in \Nbd(a)\}}{\delta_B 3^{R-1}} \right ] 
\end{align*}

This inequality might appear complicated at first but the terms can be re-arranged to give the following simple lower bound.
\begin{align*}
v_{(a, z)}(X_{(a, z)}) &\ge \sum_{b' \in B} \frac{\mathbbm{1}\{b' \in \Nbd(a)\}}{|B| \delta_B 3^{R-1}} \E_{g \sim q} \left [ \sum_{t = 1}^4 \mathbbm{1}\left \{\chi \left (x^1_{\Lambda(b)}, \dots, x^4_{\Lambda(b)} \right ) = t \right \} \, \middle |\, b = b' \right ] \\
&= \sum_{b' \in B} \frac{\mathbbm{1}\{b' \in \Nbd(a)\}}{|B| \delta_B 3^{R-1}} \E_{g \sim q} \left [ \mathbbm{1}\left \{ \chi \left ( x^1_{\Lambda(b)}, \dots, x^4_{\Lambda(b)} \right ) > 0 \right \} \, \middle |\, b = b' \right ] \\
&\ge  \sum_{b' \in B} \frac{\mathbbm{1}\{b' \in \Nbd(a)\} (1 - \vare)}{|B| \delta_B 3^{R-1}} \\
&= \frac{\delta_A (1 - \vare)}{|B| \delta_B 3^{R-1}} = \frac{1-\vare}{|A|3^{R-1}}.
\end{align*}
The second inequality follows from Observation \ref{obs:chi}. The final equality holds because the graph is both left and right regular; therefore, $|A|\delta_A = |B|\delta_B$. This completes the proof for the YES case.

\paragraph{NO Case.} Assume the underlying unique games instance comes from the NO Case. Let $X$ be any allocation where two-thirds of the agents receive large goods they value at $\frac1{\vare}$. For each $a \in A$, define $f_a:\{0, 1, 2\}^R \rightarrow \{0, 1\}$ as $f_a(x) = 1$ if $(a, x)$ is allocated a large good in the allocation $X$. Since all the large goods are allocated to agents that value them positively such that each agent receives at most one, it must be that $\E_x[f_a(x)] = 2/3$ for each $a \in A$. For each $b \in B$, define the function $f_b$ as follows:
\begin{align*}
	f_b(x) = \E_{a \in \Nbd(b)}\left [ f_a(x \circ \pi_{a, b})\right ].
\end{align*}
Note that $\E_x[f_b(x)] = 2/3$ for all $b \in B$ since $\E_x[f_a(x)] = 2/3$ for all $a \in A$. We call a node $b \in B$ {\em good} if there exists an $i \in [R]$ such that $\Inf_i^{\le d}(f_b) > \tau$, where $d$ and $\tau$ are given by Theorem \ref{thm:mossel-pairwise} with distribution $\eta'_{\vare}$ and constant $\vare$. Let $B'$ be the set of good nodes, and $S(b)$ be the set of coordinates $i$ of node $b$ with  $\Inf_i^{\le d}(f_b) > \tau$. 

\begin{claim}\label{claim:good-vertex}
$|B'| \le \vare |B|$.
\end{claim}
\begin{proof}
This follows from a standard argument; we present it here for completeness. The proof we use mirrors that of \citet[Section 4]{Bansal2010HypergraphVC}. Assume for contradiction that $|B'| > \vare|B|$. We prove that there exists a labeling $\Lambda$ that satisfies more than $\delta = \frac{\vare \tau^2}{8d}$ edges. Let $b$ be a good node and let $i$ be any coordinate in $S(b)$.
We have,
\begin{align}
\tau &\le \Inf_i^{\le d}(f_b) = \sum_{S: i \in S, |S| \le d} \E_x[f_{b, S}(x)^2] =  \sum_{S: i \in S, |S| \le d} \E_x \left [ \left (\E_{a \in \Nbd(b)}[f_{a, \pi_{b, a}(S)}(x \circ \pi_{a, b})] \right )^2 \right ] \notag\\
&\le \sum_{S: i \in S, |S| \le d} \E_x \left [\E_{a \in \Nbd(b)} \left [ \left (f_{a, \pi_{b, a}(S)}(x \circ \pi_{a, b}) \right )^2 \right ] \right ]= \E_{a \in \Nbd(b)}\left[ \Inf_{\pi_{b, a}(i)}^{\le d}(f_a)\right ]. \label{eq:high-inf}
\end{align} 
The first equality uses the Efron-Stein decomposition (Definition \ref{def:efron-stein}) and the second equality follows from the uniqueness of this decomposition.
For every $a \in A$, we define $\text{Cand}(a) = \{i \in [R] \,|\, \Inf_i^{\le d} \ge \tau/2\}$. From Lemma \ref{lem:d-inf}, $\text{Cand}(a)$ has size at most $2d/\tau$.

We construct a (random) labeling as follows. For every good node $b \in B'$, we assign it a label arbitrarily from $S(b)$. For all nodes $b \in B \setminus B'$, we assign a label arbitrarily. For all nodes $a \in A$, we pick a label uniformly at random from $\text{Cand}(a)$. If $\text{Cand}(a)$ is empty, we pick a label arbitrarily. By \eqref{eq:high-inf}, for each good node $b$ and label $i$ that we assign to it, at least $\tau/2$ fraction of its neighbors have $\Inf_{\pi_{b, a}(i)}^{\le d}(f_a) \ge \tau/2$. Therefore, for each good node $b$, we satisfy at least a $(\tau/2)(\tau/2d)$ fraction of edges adjacent to it in expectation; this uses the observation that $\text{Cand}(a)$ has size at most $2d/\tau$ for each $a \in A$. Since there are $\vare |B|$ good nodes by assumption, our random labeling satisfies at least a $\frac{\vare \tau^2}{4d}$ fraction of the edges in expectation. This shows that there exists a labeling $\Lambda$ that satisfies a $\frac{\vare \tau^2}{4d} > \delta = \frac{\vare \tau^2}{8d}$ fraction of the edges, which is a contradiction.
\end{proof}

The total utility of the agents who do not receive a large good in $X$ has two parts: utility from the dummy goods and utility from the small goods. The total utility obtained from dummy goods is at most $\vare$ (from \eqref{eq:dummy}). We can use a proof similar to the YES case to upper bound the total utility obtained from small goods. For every small good $g$, we associate it with a tuple of the form $(b, (a^1, x^1 \circ \pi_{a^1, b}), \dots (a^4, x^4 \circ \pi_{a^4, b}))$. Let $N'$ be the set of agents who do not receive a large good in $X$. For a good $g$, if $\prod_{j = 1}^4 f_{a^j}(x^j \circ \pi_{a^j, b}) = 1$, then all the agents who value the good positively are allocated large goods in $X$. Therefore, the good cannot contribute to the utility of agents in $N'$. Using this we can upper bound the utility of agents in $N'$ under the allocation $X$ as follows:

\begin{align*}
\sum_{u \in N'} v_u(X_u) &\le \vare + \sum_{g \in G} q(g) \left (1 - \prod_{j = 1}^4 f_{a^j}(x^j \circ \pi_{a^j, b}) \right ) \\
&= (1 + \vare) - \E_{g \sim q} \left [ \prod_{j = 1}^4 f_{a^j}(x^j \circ \pi_{a^j, b}) \right ] \\
&= (1 + \vare) - \sum_{b' \in B} \frac1{|B|} \E_{g \sim q} \left [ \prod_{j = 1}^4 f_{a^j}(x^j \circ \pi_{a^j, b}) \, \middle | \, b = b'\right ]
\end{align*}

In the final equality, we condition on $b$ taking a certain value, just as we did in the analysis of the YES case. Note that since $a^j$ is distributed uniformly over the neighbors of $b$ and independent of other $a^{j'}$, we can swap $f_{a^j}(x^j \circ \pi_{a^j, b})$ with $f_b(x)$. Therefore, 
\begin{align*}
\sum_{u \in N'} v_u(X_u) &\le (1 + \vare) - \sum_{b' \in B} \frac1{|B|} \E_{g \sim q} \left [ \prod_{j = 1}^4 f_{b}(x^j) \, \middle | \, b = b'\right ] \\
&= (1 + \vare) - \sum_{b' \in B} \frac1{|B|} \E_{(x^1, x^2, x^3, x^4) \sim p_{\vare}} \left [ \prod_{j = 1}^4 f_{b'}(x^j) \right ]
\end{align*} 

For any $b \in B \setminus B'$, we lower bound $\E_{(x^1, x^2, x^3, x^4) \sim p_{\vare}} \left [ \prod_{j = 1}^4 f_{b}(x^j) \right ]$ using Theorem \ref{thm:mossel-pairwise} to the value $\left ( \prod_{j = 1}^4 \E_{(x^1, x^2, x^3, x^4) \sim p_{\vare}} \left [ f_{b'}(x^j) \right ] \right ) - \vare$. For each $b \in B'$, we use the trivial lower bound of \\ $\E_{(x^1, x^2, x^3, x^4) \sim p_{\vare}} \left [ \prod_{j = 1}^4 f_{b}(x^j) \right ] \ge 0$. This simplifies our expression to:

\begin{align}
\sum_{u \in N'} v_u(X_u) &\le (1 + \vare) - \sum_{b' \in B}\frac{\mathbbm{1}\{b' \in B \setminus B'\}}{|B|} \left ( \prod_{j = 1}^4 \E_{(x^1, x^2, x^3, x^4) \sim p_{\vare}} \left [ f_{b'}(x^j)\right ] - \vare \right ) \notag\\
&\le (1 + 2\vare) - \sum_{b' \in B}\frac{\mathbbm{1}\{b' \in B \setminus B'\}}{|B|} \left ( \prod_{j = 1}^4 \E_{x} \left [ f_{b'}(x)\right ] \right ) \label{eq:neg-upper-bound}\\
&=  (1 + 2\vare) - \sum_{b' \in B}\frac{\mathbbm{1}\{b' \in B \setminus B'\}}{|B|} \left ( \left ( \frac23 \right )^4 \right ) \notag\\ 
&\le (1 + 2\vare) -  (1 - \vare) \left ( \left ( \frac23 \right )^4 \right ) 
\le \frac{65}{81} + 3\vare. \notag
\end{align}
The second inequality follows from the fact that $p_{\vare}$ is a product distribution over the balanced distribution $\eta'_{\vare}$. Therefore, $\E_{(x^1, x^2, x^3, x^4) \sim p_{\vare}} \left [ f_{b'}(x^j)\right ] = \E_x \left [f_{b'}(x) \right ] = 2/3$ for all $b' \in B$. This completes the proof for the NO Case.

\section{Nash Welfare and Budgeted Allocation}
In this section, we use the meta theorem from the previous section to prove inapproximability results for the Nash welfare and budgeted allocation objectives.


\thmnash*
\begin{proof}
We assume without loss of generality that $\vare$ is a rational constant such that $\vare < 0.01$. 
We reduce from Theorem \ref{thm:meta}. In the YES Case, there is an allocation where $2n/3$ agents receive a utility of $\frac1{\vare}$ and all the other agents receive a utility of $\frac{3(1 - \vare)}{n}$. The Nash welfare of this allocation (denoted $X^{+}$) is at least
\begin{align}
\NSW(X^+) \ge \left ( \frac1{\vare}\right )^{2/3} \left ( \frac{3(1 - \vare)}{n} \right )^{1/3}. \notag
\end{align}

In the NO Case, consider a max Nash welfare allocation $X^{-}$. In this allocation, it must be the case that all goods are allocated to an agent who values the good positively. Moreover, the large goods must be allocated such that each agent receives at most one large good. To see this, assume there exists an agent $u$ who receives two large goods it values positively and let $u'$ be any agent from the same group (as defined in Problem Instance Family \ref{prob:meta}) who does not receive a large good; note that since $u'$ does not receive a large good, $v_{u'}(X_{u'}) \le (1 + \vare)$ . Moving one large good from $u$ to $u'$ decreases the utility of $u$ by a multiplicative factor of at most $2$ but increases the utility of $u'$ by a multiplicative factor of at least $\frac1{1 + \vare} \left  (\frac1 \vare \right )$. This transfer strictly increases the Nash welfare when $\vare < 0.01$; therefore $X^{-}$ cannot allocate two large goods to the same agent. 

We have shown that $X^{-}$ satisfies the conditions required by the NO Case of Theorem \ref{thm:meta}. Let $N'$ be the set of agents who do not receive a large good in $X^{-}$. Using AM $\ge$ GM and the NO Case of Theorem \ref{thm:meta},
\begin{align*}
\prod_{u \in N'} v_u(X^{-}_u) \le \left ( \frac3n \sum_{u \in N'} v_u(X^{-}_u)\right )^{n/3} \le \left ( \frac3n \left (\frac{65}{81} + 3\vare \right )\right )^{n/3} \le \left ( \frac3n \left (\frac{65}{81} \right ) \left (1 + 5\vare \right )\right )^{n/3}.
\end{align*}

Each agent who is allocated a large good in $X^{-}$ receives a utility of at most $\frac1{\vare} + 1 + \vare \le \frac1{\vare}(1 + 5 \vare)$ since they only receive one large good. Using this, we can upper bound the Nash welfare of $X^{-}$ as
\begin{align*}
\NSW(X^{-}) \le \left ( \frac1{\vare}\right )^{2/3} \left ( \frac{3}{n} \left ( \frac{65}{81} \right ) \right )^{1/3} (1 + 5 \vare).
\end{align*}

Theorem \ref{thm:meta} says that it is NP-Hard to distinguish the YES and NO Cases (assuming UGC). Therefore, it is NP-hard for an approximation algorithm to do better than $\frac{\NSW(X^{+})}{\NSW(X^{-})}$. This ratio is at least:
\begin{align*}
	\frac{\NSW(X^{+})}{\NSW(X^{-})} \ge \left ( \frac{81}{65} \right )^{1/3} (1 - \vare)(1 - 10\vare) \ge \left ( \frac{81}{65} \right )^{1/3} - 20\vare. 	
\end{align*}

Repeating this proof with $\vare' = \vare/20$ proves the theorem.
\end{proof}

\thmbudget*
\begin{proof}
This proof is similar to the previous one.
We assume without loss of generality that $\vare$ is a rational constant such that $\vare < 0.01$. 
We reduce from Theorem \ref{thm:meta}. 
Note that we need to additionally specify the budget of each agent to complete the instance description. We set the budget of each agent to be $\frac{3(1 - \vare)}{n}$.

In the YES Case, there is an allocation (denoted $X^+$) where $2n/3$ agents receive a utility of at least $\frac1{\vare}$ and all the other agents receive a utility of $\frac{3(1 - \vare)}{n}$. Since the utility of each agent (weakly) surpasses its budget in this allocation, $\USW_b(X^{+}) = \sum_{u \in N} b_u = 3(1- \vare)$. 

In the NO Case, consider a max budgeted allocation $X$. We transform this allocation into one that  satisfies the conditions required by the NO Case of Theorem \ref{thm:meta}. From each agent's bundle in $X$, remove any goods which do not reduce the budgeted social welfare once removed. In the resulting allocation, all large goods are only allocated to agents who value them positively, and each agent receives at most one large good (the second large good would have been removed). For all the remaining large goods, we allocate them arbitrarily to agents who value them positively such that each agent is allocated at most one large good. Note that this step weakly increases the budgeted social welfare; so this new allocation (denoted $X^{-}$) maximizes the budgeted social welfare and satisfies the conditions required by the NO Case of Theorem \ref{thm:meta}.

Let $N'$ be the set of agents who do not receive a large good in $X^{-}$. Using the NO Case of Theorem \ref{thm:meta}, we can upper bound the budgeted social welfare of $X^{-}$ as
\begin{align*}
\USW_b(X^{-}) = \sum_{u \in N} \min\{b_u, v_u(X^{-}_u)\} \le \sum_{u \in N\setminus N'} b_u + \sum_{u \in N'} v_u(X^{-}_u) \le 2(1 - \vare) + \frac{65}{81} + 3\vare = \frac{227}{81} + \vare.
\end{align*}
Theorem \ref{thm:meta} says that it is NP-Hard to distinguish the YES and NO Cases (assuming UGC). Therefore, it is NP-hard for an approximation algorithm to do better than $\frac{\USW_b(X^{+})}{\USW_b(X^{-})}$. This ratio is at least:

\begin{align*}
	\frac{\USW_b(X^{+})}{\USW_b(X^{-})} \ge \frac{3(1- \vare)}{\frac{227}{81} + \vare} \ge \frac{243}{227}(1 - \vare)^2 \ge \frac{243}{227} - 4\vare.
\end{align*}

Repeating this proof with $\vare' = \vare/4$ proves the theorem.
\end{proof}

%

\section{Generalized Assignment Problem}\label{sec:gap}
In this section, we prove an inapproximability result for the max generalized assignment problem (GAP). The previous section used instances from Problem Instance Family \ref{prob:meta} to show inapproximability. However, these instances cannot be used to show an inapproximability result for the max GAP problem. This is because each large good has a very large value compared to the small goods; therefore, the contribution of the small goods to the utilitarian social welfare is very small. Due to this, the max utilitarian social welfare of the YES and NO Cases of Theorem \ref{thm:meta} are very close to each other. 

The solution to this issue is reducing the value of each large good to create a gap between the YES and NO Cases of Theorem \ref{thm:meta}. However, once the value of each large good is reduced, we can no longer guarantee that all large goods are allocated in the optimal allocation. This means we cannot directly apply Theorem \ref{thm:meta}. Nevertheless, we show that a careful modification to the proof of Theorem \ref{thm:meta} allows us to prove an inapproximability result. 

\thmgap*
\begin{proof}
We assume without loss of generality that $\vare$ is a rational constant such that $\vare < 0.01$. 
This proof uses the same notation as the proof of Theorem \ref{thm:meta}.
We reduce from the unique games problem. Let $\mathcal{L} = (A \cup B, E, \Pi, [R])$ be an instance from Conjecture \ref{conj:unique-games} with $\delta = \frac{\vare \tau^2}{8d}$, where $\tau$ and $d$ are obtained from Theorem \ref{thm:mossel-pairwise} with distribution $\eta'_{\vare}$ and constant $\vare$. 

We construct our GAP instance the exact same way as Theorem \ref{thm:meta}. The only modification we make is that all large goods are valued at $\frac{c}{|A|3^R}$ with $c = \frac{32}{27}$ by the agents who value them positively, as opposed to $\frac1{\vare}$. The value of $c$ is chosen specifically to maximize the inapproximability factor. To complete the description of the GAP instance, we need to define good sizes and agent capacities. We define the size of each large good as $1$ and the capacity of each agent as $1$. For all the other goods, we give them a positive size small enough that any bundle that does not contain a large good has size at most $1$. Essentially, any allocation which allocates at most one large good to each agent satisfies the capacity constraints. 

\noindent \textbf{(YES Case):} Given an instance of the YES Case of Conjecture \ref{conj:unique-games}, consider the allocation constructed by Theorem \ref{thm:meta} in the YES Case. This allocation satisfies all the capacity constraints once we remove any small goods from the bundles of agents who receive large goods. Let us use $X^+$ to denote this modified allocation. Let $N'$ be the set of agents not allocated large goods in $X^{+}$. We can easily lower bound the $\USW$ of this allocation:
\begin{align*}
\USW(X^+) = \sum_{u \in N'} v_u(X^{+}_u) + \sum_{u \in N \setminus N'} v_u(X^{+}_u) \ge \frac{1 - \vare}{|A| 3^{R-1}} \cdot (|A|3^{R-1}) + \frac{c}{|A|3^R} \cdot (2|A|3^{R-1}) = \frac{145}{81} - \vare.
\end{align*}

\noindent \textbf{(NO Case):} 
Assume the underlying unique games instance comes from the NO Case. Let $X$ be any feasible allocation. For each $a \in A$, define $f_a:\{0, 1, 2\}^R \rightarrow \{0, 1\}$ as $f_a(x) = 1$ if $(a, x)$ is allocated a large good in the allocation $X$ that it values at $\frac{c}{|A|3^R}$; since the allocation $X$ is feasible, each agent receives at most one such good. For each $b \in B$, define the function $f_b$ as follows:
\begin{align*}
	f_b(x) = \E_{a \in \Nbd(b)}\left [ f_a(x \circ \pi_{a, b})\right ].
\end{align*}
We call a node $b \in B$ {\em good} if there exists an $i \in [R]$ such that $\Inf_i^{\le d}(f_b) > \tau$, where $d$ and $\tau$ are given by Theorem \ref{thm:mossel-pairwise} with distribution $\eta'_{\vare}$ and constant $\vare$. Let $B'$ be the set of good nodes. We can prove the following claim using the same proof as Claim \ref{claim:good-vertex}.

\begin{claim}\label{claim:gap-good-vertex}
$|B'| \le \vare |B|$.
\end{claim}

We also need the following useful claim.

\begin{claim}\label{claim:gap-a-b}
$\E_{b \in B} \left [\E_x [f_b(x)] \right ] = \E_{a \in A} \left [\E_x [f_a(x)] \right ]$.
\end{claim}
\begin{proof}
This is true due to the following sequence of equalities.
\begin{align*}
\E_{b \in B} \left [\E_x [f_b(x)] \right ] &= \E_{b \in B} \left [\E_x \left [ \E_{a \in \Nbd(b)}[f_a(x \circ \pi_{a, b})] \right ] \right ] = \E_{b \in B} \left [\E_{a \in \Nbd(b)} \left [ \E_{x}[f_a(x \circ \pi_{a, b})] \right ] \right ] \\
&= \E_{b \in B} \left  [\E_{a \in \Nbd(b)} \left [ \E_{x}[f_a(x)] \right ] \right ] = \E_{a \in A} \left [ \E_x [f_a(x)] \right ].
\end{align*}
The penultimate equality holds because sampling an $x$ uniformly at random and then applying a permutation to $x$ is equivalent to just sampling $x$ uniformly at random. The final inequality holds since the graph is left and right regular; therefore, sampling $b \in B$ and then sampling one its neighbors $a$ is equivalent to sampling $a \in A$ uniformly at random.
\end{proof}

Let $N'$ be the set of agents who do not receive a large good they value at $\frac{c}{|A|3^R}$ in $X$. We can upper bound the USW of $X$ using a similar argument to Theorem \ref{thm:meta}. We get the following upper bound using \eqref{eq:neg-upper-bound}; all previous steps that lead to \eqref{eq:neg-upper-bound} can be replicated here as well, but are omitted.
\begin{align*}
\USW(X) &= \sum_{u \in N'} v_u(X_u) + \sum_{u \in N \setminus N'} v_u(X_u) \\
&\le (1 + 2\vare) - \sum_{b' \in B}\frac{\mathbbm{1}\{b' \in B \setminus B'\}}{|B|} \left ( \prod_{j = 1}^4 \E_{x} \left [ f_{b'}(x)\right ] \right ) + \sum_{a \in A} \sum_{x \in \{0, 1, 2\}^R} \frac{c}{|A|3^{R}} f_a(x) \\
&= (1 + 2\vare) - \sum_{b' \in B}\frac{\mathbbm{1}\{b' \in B \setminus B'\}}{|B|} \left (\E_{x} \left [ f_{b'}(x)\right ]\right )^4 + c\E_{a \in A} [\E_{x}[f_a(x)]] \\
&= (1 + 2\vare) - \sum_{b' \in B}\frac{\mathbbm{1}\{b' \in B \setminus B'\}}{|B|} \left (\E_{x} \left [ f_{b'}(x)\right ]\right )^4 + c\E_{b \in B} [\E_{x}[f_b(x)]] 
\end{align*}

In the final equality, we use Claim \ref{claim:gap-a-b}. We can further re-arrange terms to get

\begin{align*}
\USW(X) &\le (1 + 2\vare) - \sum_{b' \in B}\frac{\mathbbm{1}\{b' \in B \setminus B'\}}{|B|} \left ( \left (\E_{x} \left [ f_{b'}(x)\right ]\right )^4 - c\E_x[f_{b'}(x)] \right ) + c\sum_{b \in B'} \frac1{|B|} \E_x[f_b(x)] \\ 
&\le (1 + 2\vare) - \sum_{b' \in B}\frac{\mathbbm{1}\{b' \in B \setminus B'\}}{|B|} \left ( \left (\E_{x} \left [ f_{b'}(x)\right ]\right )^4 - c\E_x[f_{b'}(x)] \right ) + 2\vare 
\end{align*}

In the second inequality, we use $|B'| \le \vare|B|$ (Claim \ref{claim:gap-good-vertex}) and $\E_x[f_b(x)] \le 1$. The polynomial $x^4 - cx$ has a minimum value of $-\frac{48}{81}$. We can use this to simplify our bounds further. 

\begin{align*}
\USW(X) &\le  (1 + 4\vare) + \sum_{b' \in B}\frac{\mathbbm{1}\{b' \in B \setminus B'\}}{|B|} \cdot \frac{48}{81} \\
&\le  (1 + 4\vare) + \frac{48}{81} \le \frac{129}{81} + 5\vare.
\end{align*}

Let $X^{-}$ be the feasible allocation which maximizes USW in the NO Case. It is NP-hard to distinguish between YES and NO Cases (assuming UGC). Therefore, it is NP-hard for an approximation algorithm to do better than
\begin{align*}
\frac{\USW(X^{+})}{\USW({X^{-}})} \ge \frac{\frac{145}{81} - \vare}{\frac{129}{81} + 5\vare} \ge \frac{145}{129}(1 - 5\vare) - \vare \ge \frac{145}{129} - 11\vare.
\end{align*}
 
Repeating the above proof with $\vare' = \vare/11$ proves the theorem.
\end{proof}

\section{Conclusion and Future Work}
In this work, we introduce a novel dictator test for the indivisible good allocation problem. We use this test to prove inapproximability results for several fundamental objectives in the indivisible good allocation literature, such as max Nash welfare, max budgeted allocation and the max generalized assignment problem. The generality of our proof technique suggests that it potentially has broader applications beyond those considered in this paper. 

A key element of our dictator test (and therefore, our inapproximability results) is the distribution $\eta$. While $\eta$ certainly satisfies several desirable properties, it is unclear whether $\eta$ is the best distribution for a dictator test using our problem. It is possible that other pairwise independent distributions lead to improved inapproximability results. This is a promising direction for future work.

\section*{Acknowledgements}
The author would like to thank Rohit Vaish for pushing him to write this paper.
The author is funded by the National Science Foundation (NSF) Career Award 2441296 and Grant RI-2327057.
\bibliographystyle{plainnat}
\bibliography{abb,references.bib}

@STRING{acm = {ACM Press}}

@STRING{approx = { International Workshop on Approximation Algorithms for Combinatorial
	Optimization Problems (APPROX)}}

@STRING{ec = { ACM Conference on Economics and Computation (EC)}}

@STRING{focs = { Symposium on Foundations of Computer Science (FOCS)}}

@STRING{icalp = { International Colloquium on Automata, Languages and Programming (ICALP)}}

@STRING{itcs = { Innovations in Theoretical Computer Science Conference (ITCS)}}

@STRING{proc = {Proceedings of the }}

@STRING{soda = { Annual ACM-SIAM Symposium on Discrete Algorithms (SODA)}}

@STRING{springer = {Springer-Verlag}}

@STRING{stoc = { Annual ACM Symposium on Theory of Computing (STOC)}}

@STRING{swat = { Scandinavian Workshop on Algorithm Theory (SWAT)}}

@inproceedings{Cole2017nashwelfare,
	author = {Cole, Richard and Devanur, Nikhil and Gkatzelis, Vasilis and Jain, Kamal and Mai, Tung and Vazirani, Vijay V. and Yazdanbod, Sadra},
	booktitle = {Proceedings of the 2017 ACM Conference on Economics and Computation},
	numpages = {2},
	pages = {459--460},
	title = {Convex Program Duality, Fisher Markets, and Nash Social Welfare},
	year = {2017}}

@article{Barman2018FindingFA,
	author = {Siddharth Barman and Sanath Kumar Krishnamurthy and Rohit Vaish},
	journal = proc # {19th} # ec,
	pages = {557--574},
	title = {Finding Fair and Efficient Allocations},
	year = {2018}}

@inproceedings{Azar2008BudgetedAllocation,
	author = {Azar, Yossi and Birnbaum, Benjamin and Karlin, Anna R. and Mathieu, Claire and Nguyen, C. Thach},
	booktitle = proc # {35th} # icalp,
	pages = {186--197},
	title = {Improved Approximation Algorithms for Budgeted Allocations},
	year = {2008}}

@inproceedings{Srinivasan2008Budgeted,
	author = {Srinivasan, Aravind},
	booktitle = proc # {11th} # approx,
	pages = {247--253},
	title = {Budgeted Allocations in the Full-Information Setting},
	year = {2008}}

@inproceedings{Garg2001BudgetedAllocations,
	author = {Garg, Rahul and Kumar, Vijay and Pandit, Vinayaka},
	booktitle = proc # {4th} # approx,
	pages = {102--113},
	title = {Approximation Algorithms for Budget-Constrained Auctions},
	year = {2001}}

@inproceedings{Andelman2004BudgetedAllocations,
	author = {Andelman, Nir and Mansour, Yishay},
	booktitle = {Proceedings of the 9th Scandinavian Workshop on Algorithm Theory (SWAT)},
	pages = {26--38},
	title = {Auctions with Budget Constraints},
	year = {2004}}

@inproceedings{Kalaitzis2016ImprovedBudgetedAllocation,
	author = {Kalaitzis, Christos},
	booktitle = proc # {27th} # soda,
	pages = {1048--1066},
	title = {An Improved Approximation Guarantee for the Maximum Budgeted Allocation Problem},
	year = {2016}}

@article{Kalaitzis2015ConfigurationLPBudgetedAllocation,
	author = {Kalaitzis, Christos and Madry, Aleksander and Newman, Alantha and Pol\'a\v{c}ek, Luk\'a\v{s} and Svensson, Ola},
	journal = {Mathematical Programming},
	number = {1-2},
	pages = {427--462},
	publisher = {Springer},
	series = {Series B},
	title = {On the Configuration LP for Maximum Budgeted Allocation},
	volume = {154},
	year = {2015}}

@inproceedings{garg2018budgetadditive,
	author = {Garg, Jugal and Hoefer, Martin and Mehlhorn, Kurt},
	booktitle = proc # {29th} # soda,
	numpages = {15},
	pages = {2326--2340},
	title = {Approximating the nash social welfare with budget-additive valuations},
	year = {2018}}

@article{Caragiannis2019MNW,
	articleno = {12},
	author = {Caragiannis, Ioannis and Kurokawa, David and Moulin, Herv\'{e} and Procaccia, Ariel D. and Shah, Nisarg and Wang, Junxing},
	journal = {ACM Trans. Econ. Comput.},
	month = sep,
	number = {3},
	numpages = {32},
	title = {The Unreasonable Fairness of Maximum Nash Welfare},
	volume = {7},
	year = {2019}}

@misc{Bei2025MNW,
	archiveprefix = {arXiv},
	author = {Xiaohui Bei and Yuda Feng and Yang Hu and Shi Li and Ruilong Zhang},
	eprint = {2504.09669},
	primaryclass = {cs.GT},
	title = {Nash Social Welfare with Submodular Valuations: Approximation Algorithms and Integrality Gaps},
	year = {2025}}

@article{Chakrabarty2010BudgetAdditive,
	author = {Chakrabarty, Deeparnab and Goel, Gagan},
	journal = {SIAM J. Comput.},
	number = {6},
	numpages = {23},
	pages = {2189--2211},
	publisher = {Society for Industrial and Applied Mathematics},
	title = {On the Approximability of Budgeted Allocations and Improved Lower Bounds for Submodular Welfare Maximization and GAP},
	volume = {39},
	year = {2010}}

@inproceedings{Anari2017MNW,
	author = {Nima Anari and Shayan Oveis Gharan and Amin Saberi and Mohit Singh},
	booktitle = proc # {8th} # itcs,
	pages = {36:1--36:12},
	title = {Nash Social Welfare, Matrix Permanent, and Stable Polynomials},
	year = {2017}}

@inproceedings{cole2015nash,
	author = {Cole, Richard and Gkatzelis, Vasilis},
	booktitle = proc # {47th} # stoc,
	pages = {371--380},
	title = {Approximating the Nash Social Welfare with Indivisible Items},
	year = {2015}}

@article{lee:j:apx-hardness,
	author = {Euiwoong Lee},
	journal = {Information Processing Letters},
	pages = {17-20},
	title = {{APX}-Hardness of Maximizing {N}ash Social Welfare with Indivisible Items},
	volume = {22},
	year = {2017}}

@article{amanatidis2023fairdivisionsurvey,
	author = {Georgios Amanatidis and Haris Aziz and Georgios Birmpas and Aris Filos-Ratsikas and Bo Li and Herv{\'e} Moulin and Alexandros A. Voudouris and Xiaowei Wu},
	journal = {Artificial Intelligence},
	pages = {103965},
	title = {Fair division of indivisible goods: Recent progress and open questions},
	volume = {322},
	year = {2023}}

@incollection{Morales2004GAPSurvey,
	address = {Boston, MA, USA},
	author = {Dolores Romero Morales and H. Edwin Romeijn},
	booktitle = {Handbook of Combinatorial Optimization, Supplement Volume B},
	editor = {Ding-Zhu Du and Panos M. Pardalos},
	pages = {259--311},
	publisher = {Springer},
	title = {The Generalized Assignment Problem and Extensions},
	year = {2004}}

@article{Oncan2007GAPSurvey,
	author = {Temel {\"O}ncan},
	journal = {INFOR: Information Systems and Operational Research},
	number = {3},
	pages = {123--141},
	title = {A Survey of the Generalized Assignment Problem and Its Applications},
	volume = {45},
	year = {2007}}

@inproceedings{Feige2006GAPImprovement,
  author={Feige, Uriel and Vondrak, Jan},
  booktitle=proc # {47th} # focs, 
  title={Approximation algorithms for allocation problems: Improving the factor of 1 - 1/e}, 
  year={2006},
  pages={667-676}}

@inproceedings{Fleischer2006GAP,
	author = {Lisa Fleischer and Michel X. Goemans and Vahab Mirrokni and Maxim Sviridenko},
	booktitle = proc # {16th} # soda,
	pages = {611--620},
	title = {Tight Approximation Algorithms for Maximum General Assignment Problems},
	year = {2006}}

@inproceedings{Chekuri2000MultipleKnapsack,
	author = {Chekuri, Chandra and Khanna, Sanjeev},
	booktitle = proc # {11th} # soda,
	pages = {213--222},
	title = {A PTAS for the multiple knapsack problem},
	year = {2000}}

@article{ShmoysTardos1993GAP,
	author = {David B. Shmoys and {\'E}va Tardos},
	journal = {Mathematical Programming},
	number = {1-3},
	pages = {461--474},
	publisher = {Springer},
	title = {An Approximation Algorithm for the Generalized Assignment Problem},
	volume = {62},
	year = {1993}}

@article{Hastad2001OptimalInapproximability,
	author = {H{\aa}stad, Johan},
	journal = {Journal of the ACM},
	number = {4},
	pages = {798--859},
	title = {Some Optimal Inapproximability Results},
	volume = {48},
	year = {2001}}

@article{Arora1998ProofVerification,
	author = {Arora, Sanjeev and Lund, Carsten and Motwani, Rajeev and Sudan, Madhu and Szegedy, Mario},
	journal = {Journal of the ACM},
	number = {3},
	pages = {501--555},
	title = {Proof Verification and Hardness of Approximation Problems},
	volume = {45},
	year = {1998}}

@article{AroraSafra1998PCP,
	author = {Arora, Sanjeev and Safra, Shmuel},
	journal = {Journal of the ACM},
	number = {1},
	pages = {70--122},
	title = {Probabilistic Checking of Proofs: A New Characterization of NP},
	volume = {45},
	year = {1998}}

@inproceedings{Khot2002UniqueGames,
	author = {Khot, Subhash},
	booktitle = proc # {34th} # stoc,
	pages = {767--775},
	title = {On the Power of Unique 2-Prover 1-Round Games},
	year = {2002}}

@article{KhotRegev2008VertexCover,
	author = {Khot, Subhash and Regev, Oded},
	journal = {Journal of Computer and System Sciences},
	number = {3},
	pages = {335--349},
	title = {Vertex Cover Might Be Hard to Approximate to Within $2 - \varepsilon$},
	volume = {74},
	year = {2008}}

@article{Khot2007MaxCutUGC,
	author = {Khot, Subhash and Kindler, Guy and Mossel, Elchanan and O'Donnell, Ryan},
	journal = {SIAM Journal on Computing},
	number = {1},
	pages = {319--357},
	title = {Optimal Inapproximability Results for {MAX-CUT} and Other 2-Variable CSPs?},
	volume = {37},
	year = {2007}}

@inproceedings{Raghavendra2008CSP,
	author = {Raghavendra, Prasad},
	booktitle = proc # {40th} # stoc,
	numpages = {10},
	pages = {245--254},
	title = {Optimal algorithms and inapproximability results for every CSP?},
	year = {2008}}

@inproceedings{Bansal2009OneFreeBit,
	author = {Bansal, Nikhil and Khot, Subhash},
	booktitle = proc # {50th} # focs,
	pages = {453-462},
	title = {Optimal Long Code Test with One Free Bit},
	year = {2009}}

@inproceedings{Svensson2010Precedence,
	author = {Svensson, Ola},
	booktitle = proc # {42nd} # stoc,
	numpages = {10},
	pages = {745--754},
	title = {Conditional hardness of precedence constrained scheduling on identical machines},
	year = {2010}}

@inproceedings{Bansal2010HypergraphVC,
	author = {Bansal, Nikhil and Khot, Subhash},
	booktitle = proc # {37th} # icalp,
	pages = {250--261},
	title = {Inapproximability of Hypergraph Vertex Cover and Applications to Scheduling Problems},
	year = {2010}}

@article{mossel2010gaussian,
	author = {Mossel, Elchanan},
	journal = {Geometric and Functional Analysis},
	number = {6},
	pages = {1713--1756},
	title = {Gaussian Bounds for Noise Correlation of Functions},
	volume = {19},
	year = {2010}}

@book{odonnell2014booleanfunctions,
	address = {Cambridge, UK},
	author = {O'Donnell, Ryan},
	publisher = {Cambridge University Press},
	title = {Analysis of Boolean Functions},
	year = {2014}}

@article{Austrin2009PairwiseIndependence,
	author = {Per Austrin and Elchanan Mossel},
	journal = {Computational Complexity},
	number = {2},
	pages = {249--271},
	title = {Approximation Resistant Predicates from Pairwise Independence},
	volume = {18},
	year = {2009}}

\appendix

\section{Generalizations of Our Dictator Test}\label{sec:generalizations}
In this section, we describe a family of dictator tests that are in some sense, generalizations of the dictator test in Section \ref{sec:dictator-test}. The purpose of this section is to provide some intuition into the specific choices behind our dictator test.

Consider some function $f: [q]_0^R \rightarrow \{0, 1\}$ for positive integers $R$ and $q$ such that $(q+1)$ is prime. We construct a dictator test for this function similar to Section \ref{sec:dictator-test}. We define a probability distribution $\eta_q: [q]_0^{q + 2} \rightarrow [0, 1]$ as follows: $\eta_{q}$ is the uniform distribution over all tuples of the form $(a, b, a + b, a + 2b, \dots, a + qb)$. Here, addition is defined modulo $(q + 1)$. $\eta_q$ satisfies the same properties that $\eta$ does in Claim \ref{claim:eta-properties}. We can define the slightly noisy distribution $\eta'_{(q, \vare)}: [q]_0^{q + 2} \rightarrow [0, 1]$ as $\eta'_{(q, \vare)} = (1 - \vare)\eta_{q} + \vare \mu$, where $\mu$ is the uniform distribution over $[q]^{q + 2}_0$. $\eta'_{(q, \vare)}$ satisfies a variant of Claim \ref{claim:eta-prime-properties}. 

Let $p_{(q, \vare)}$ denote a probability distribution over tuples of $q+2$ points $(x^1, \dots, x^{(q + 2)})$ in $[q]_0^{R}$ such that $p_{(q, \vare)}(x^1, \dots, x^{(q + 2)}) = \prod_{i \in [R]} \eta'_{(q, \vare)}(x^1_i, \dots, x^{(q+2)}_i)$. We can use the distribution $p_{(q, \vare)}$ to design a dictator test similar to Section \ref{sec:dictator-test}.

\begin{prob}\label{prob:general-dictator-test}
We have a set of $(q + 1)^R$ agents, one for every element in $[q]_0^R$. For every tuple of $q+2$ agents $(x^1, \dots, x^{(q+2)})$ we define one good which is valued at $p_{(q, \vare)}(x^1, \dots, x^{(q+2)})$ by the agents $x^1, \dots, x^{q+2}$, and valued at $0$ by all the other agents. We refer to these goods as {\em small} goods. We define a set of $q \cdot {(q + 1)}^{R-1}$ large goods that is valued by all the agents at $\frac1{\vare}$.   
\end{prob}

For every allocation $X$ in this instance, we define the boolean function $f_X: [q]_0^R \rightarrow \{0, 1\}$ as $f_X(x) = 1$ if and only if $x$ receives a large good in the allocation $X$. We can prove a theorem of the following form (proof omitted).

\begin{theorem}[General Dictator Test]\label{thm:general-dictator-test}
In an instance from Problem Instance Family \ref{prob:general-dictator-test}, let $X$ be an allocation where $q \cdot (q+1)^{R-1}$ agents receive a large good, and let $N'$ be the set of agents which do not receive any large goods. Then:
\begin{description}[font=\normalfont]
    \item[\textsc{Completeness.}] If $f_X$ is defined as $f_X(x) = \mathbbm{1}\{x_i > 0\}$ for some $i \in [R]$, then there is an allocation $X'$ with $f_X = f_{X'}$ such that for all agents in $u \in N'$, $v_u(X'_u) \ge \frac{1}{(q + 1)^{R-1}}(1 - \vare)$.
    \item[\textsc{Soundness.}] There exists a $d > 0$ and $\tau > 0$ such that if $f_X$ has $\Inf^{\le d}_i(f_X) \le \tau$ for all $i \in [R]$, then $\sum_{u \in N'} v_u(X_u) \le 1 - \left ( \frac{q}{q+1} \right )^{q+2} + \vare$. 
\end{description}
\end{theorem}

When $q$ is high, most agents receive large goods, and this leads to a low inapproximability factor. However, when $q$ is high, the soundness property is stronger; agents who do not receive a large good receive very low utility. This creates a trade-off. In all the problems we consider, the inapproximability factor is maximized when $q = 2$. Surprisingly, when $q = 1$, we recover the previous best inapproximability factor for all the problems we consider.  

\end{document}